\documentclass[11pt,letterpaper]{article}

\usepackage{geometry}

\usepackage[numbers]{natbib}
\usepackage{optparams,ifthen,xspace}

\makeatletter
\let\citeptemp\citep
\long\def\citep@[#1][#2]#3{%
	\ifthenelse{\equal{#1}{}}{%
		\ifthenelse{\equal{#2}{}}{%
			\citeptemp[][]{#3}}{%
			\citeptemp[][#2]{#3}}}{%
		\ifthenelse{\equal{#2}{}}{%
			(#1\citeptemp[][]{#3})}{%
			(#1\citeptemp[][]{#3}, #2)}}}
\renewcommand{\citep}{\optparams{\citep@}{[][]}}
\makeatother

\usepackage{amsmath,amssymb,amsthm}
\theoremstyle{plain}
\newtheorem{proposition}{Proposition}
\newtheorem{theorem}{Theorem}
\newtheorem{lemma}{Lemma}
\newtheorem{corollary}{Corollary}

\theoremstyle{definition}
\newtheorem{example}{Example}

\def\clap#1{\hbox to 0pt{\hss#1\hss}}
  \def\mathclap{\mathpalette\mathclapinternal}

\def\mathclapinternal#1#2{\clap{$\mathsurround=0pt#1{#2}$}}

\newcommand{\ie}{i.e.,\xspace}
\newcommand{\eg}{e.g.,\xspace}

\newcommand{\reals}{\mathbb{R}}

\newcommand{\hX}{\hat{X}}
\newcommand{\hx}{\hat{x}}
\newcommand{\hf}{\hat{f}}
\newcommand{\hp}{\hat{p}}
\newcommand{\hM}{\hat{M}}

\newcommand{\Wmax}{W_{\mathit{max}}}

\newcommand{\ctrf}{\reals_{>}^n}
\newcommand{\midd}{:}
\newcommand{\aGSP}{$\alpha$-GSP\xspace}
\newcommand{\one}{\mathbf{1}}
\newcommand{\oneGSP}{$\one$-GSP\xspace}  
\newcommand{\aVCG}{$\alpha$-VCG\xspace}
\newcommand{\nVCG}{$n$-VCG\xspace}

\newcommand{\SigmaVCG}{$\Sigma$-VCG\xspace}

\newcommand{\secref}[1]{Section~\ref{#1}}

\newcommand{\lemref}[1]{Lemma~\ref{#1}}
\newcommand{\thmref}[1]{Theorem~\ref{#1}}
\newcommand{\appref}[1]{Appendix~\ref{#1}}

\newcommand{\propref}[1]{Proposition~\ref{#1}}

\begin{document}

\title{Simplicity-Expressiveness Tradeoffs in Mechanism Design} 


\author{%
	Paul D\"{u}tting\thanks{%
	Ecole Polytechnique F\'ed\'erale de Lausanne (EPFL), 
	Station 14, CH-1015 Lausanne, Switzerland, 
	Email: \texttt{paul.duetting@epfl.ch}. 
	Research supported by a European Young Investigators Award (EURYI).}
	\and 
	Felix Fischer\thanks{%
	School of Engineering and Applied Sciences, Harvard University, 
	33 Oxford Street, Cambridge, MA~02138, USA, 
	Email: \texttt{fischerf@seas.harvard.edu}. 
	Research supported by DFG grant FI~1664/1-1.} 
	\and David C.~Parkes\thanks{%
	School of Engineering and Applied Sciences, Harvard University, 
	33 Oxford Street, Cambridge, MA~02138, USA, 
	Email: \texttt{parkes@eecs.harvard.edu}}}

\date{}

\maketitle

\begin{quote}
	\emph{Everything should be made as simple as possible, but no simpler.} \\\hspace*{\fill} 
	-- Albert Einstein 
\end{quote}

\begin{abstract}
A fundamental result in mechanism design theory, the so-called revelation principle, asserts that for many questions concerning the existence of mechanisms with a given outcome one can restrict attention to truthful direct revelation-mechanisms. In practice, however, many mechanism use a restricted message space. This motivates the study of the tradeoffs involved in choosing simplified mechanisms, which can sometimes bring benefits in precluding bad or promoting good equilibria, and other times impose costs on welfare and revenue. We study the simplicity-expressiveness tradeoff in two representative settings, sponsored search auctions and combinatorial auctions, each being a canonical example for complete information and incomplete information analysis, respectively. We observe that the amount of information available to the agents plays an important role for the tradeoff between simplicity and expressiveness. 
\end{abstract}

\section{Introduction}

A fundamental result in mechanism design theory, the so-called revelation principle, asserts that for many questions concerning the existence of mechanisms with a given outcome one can restrict attention to truthful direct-revelation mechanisms, \ie mechanisms in which agents truthfully reveal their type by stating it directly. In practice, however, the assumptions underlying the revelation principle often fail: other equilibria may exist besides the truthful one, or computational limitations may interfere. 
As a consequence, many practical mechanisms use a restricted message space. This motivates the study of the tradeoffs involved in choosing among mechanisms with different degrees of expressiveness. Despite their practical importance, these tradeoffs are currently only poorly understood. 

In sponsored search auctions, and adopting a complete information analysis, allowing every agent~$i$ to submit a valuation $v_{ij}$ for each slot~$j$ means that both the Vickrey-Clarke-Groves (VCG) mechanism and the Generalized Second Price (GSP) mechanism always admit a Nash equilibrium with zero revenue~\cite{Milg10a}. If agents face a small cost for submitting a non-zero bid, this becomes the unique equilibrium. If instead agent~$i$ is asked for a single bid~$b_i$, and his bid for slot $j$ is derived by multiplying it with a slot-specific click-through rate $\alpha_j$, the zero-revenue equilibria are eliminated. More surprisingly, this simplification does not introduce new equilibria (even when $\alpha_j$ is not correct for every agent~$i$), so minimum revenue over all Nash equilibria is strictly greater for the simplification than for the original mechanism. Moreover, if valuations can actually be decomposed into an agent-specific valuation $v_i$ and click-through rates $\alpha_j$, the simplification still has an efficient Nash equilibrium. \citet{Milg10a} concluded that simplification can be beneficial and need not come at a cost. 

For combinatorial auctions, and adopting an incomplete information analysis, the maximum social welfare over all outcomes of a mechanism strictly increases with expressiveness, for a particular measure of expressiveness based on notions from computational learning theory~\citep{BSS08}. Implicit in the result of~\citet{BSS08} is the conclusion that more expressiveness is generally desirable, as it allows a mechanism to achieve a more efficient outcome in more instances of the problem. 

Each of these results tells only part of the story. While \citeauthor{Milg10a}'s results on the benefits of simplicity are developed within an equilibrium framework and can be extended beyond sponsored search auctions, they critically require settings with complete information amongst agents in order to preclude bad equilibria while retaining good ones. In particular, \citeauthor{Milg10a} does not consider the potential loss in efficiency or revenue that can occur when agents are ignorant of each other's valuations when deciding how to bid within a simplified bidding language. \citeauthor{BSS08}, on the other hand, develop their results on the benefits of expressiveness (and thus the cost of simplicity) in an incomplete information context, but largely in the absence of equilibrium considerations.\footnote{The equilibrium analysis that \citeauthor{BSS08} provide is in regard to identifying a particular mechanism design in which the maximum social welfare achievable in any outcome can be achieved in a particular Bayes-Nash equilibrium.} In particular, these authors do not consider the potential problems that can occur due to the existence of bad equilibria in expressive mechanisms. 

\paragraph{Our Contribution}
The contribution of this paper is two\-fold. On a conceptual level, we analyze how different properties of a simplification affect the set of equilibria of a mechanism in both complete and incomplete information settings and argue that well-chosen simplifications can have a positive impact on the set of equilibria as a whole; either by precluding undesirable equilibria or by promoting desirable equilibria. We thus extend Milgrom's emphasis on simplification as a tool that enables equilibrium selection. On a technical level, we analyze simplified mechanisms for sponsored search auctions with complete information and combinatorial auctions with incomplete information. 

An important property when analyzing the impact of simplification on the set of equilibria is \emph{tightness}~\citep{Milg10a}, which requires that no additional equilibria are introduced. We observe that tightness can be achieved equally well in complete and incomplete information settings and give a sufficient condition.\footnote{This condition was already considered by~\citeauthor{Milg10a}, but only in the complete information case.}
Complementary to tightness is a property we call \emph{totality}, which requires all equilibria of the original mechanism to be preserved.
To the end of equilibrium selection, totality needs to be relaxed. Particular relaxations we consider in this paper are the preservation of the VCG outcome, \ie the outcome obtained in the dominant strategy equilibrium of the fully expressive VCG mechanism, and the existence of an equilibrium with a certain amount of social welfare or revenue relative to the VCG outcome. In addition, one might require that the latter property holds for \emph{every} equilibrium of the simplified mechanism. 

In the context of sponsored search, one reason to prefer a simplification is to preclude the zero revenue equilibrium discussed above. Another interesting property of a simplified GSP mechanism is that it preserves the VCG outcome even when the assumed click-through rates $\alpha_j$ are inexact. Recognizing that this claim cannot be made for the VCG mechanism under the same simplification,~\citet{Milg10a} uses this as an argument for the superiority of the GSP mechanism. But, this result that GSP is Vickrey-preserving requires an unnatural condition on the relationship between the assumed click-through rates and prices and thus agents' bids, and moreover does not preclude alternate simplifications of VCG that succeed in being Vickrey-preserving. In addition, we observe that the simplifications can still suffer arbitrarily low revenue in some equilibrium, in comparison with the VCG revenue. 
In our results on sponsored search, we identify a simplified GSP mechanism that preserves the VCG outcome without requiring any knowledge of the actual click-through rates, precludes zero revenue equilibrium, and always recovers at least half of the VCG payments for all slots but the first. For simplified VCG mechanisms, we obtain a strong negative result: \emph{every} simplification of VCG that supports the (efficient) VCG outcome in some equilibrium also has an equilibrium in which revenue is arbitrarily smaller than in the VCG outcome. 

In the context of combinatorial auctions, paradigmatic of course of settings with incomplete information, we first recall the previous observations by \citeauthor{HKMT04}~\citep{HM04,HKMT04} in regard to the existence of multiple, non-truthful, \emph{ex post} Nash equilibria of the VCG mechanism, each of which offers different welfare and revenue properties. In particular, if one assumes that participants will select equilibria with a particular (maximum) number of bids, social welfare can differ greatly among the different equilibria, revenue can be zero for some of them, and the existence of multiple Pareto optimal equilibria can make equilibrium selection for participants hard to impossible. 
Focusing again on \emph{tight} simplifications, we connect the analysis of \citeauthor{HKMT04}~\citep{HM04,HKMT04} with tightness, by establishing that a simplification is tight if and only if bids are restricted to a subset~$\Sigma$ of the bundles with a \emph{quasi-field} structure~\cite{HM04,HKMT04}, with values for the other bundles derived as the maximum value of any contained bundle. Through insisting on a tight simplification, we ensure that the worst-case behavior is no worse than that of the fully expressive VCG mechanism, even when $\Sigma$ (although simplified) contains too many bundles for agents to bid on all of them. Moreover, using a quasi-field simplification ensures that agents do not experience regret with respect to the bidding language, in the sense of wanting to send a message \emph{ex post} that was precluded. 
Finally, as any restriction of the bids to a subset of the bundles, restricting the bids to a quasi-field $\Sigma$ makes it a dominant strategy equilibrium for the agents to bid truthfully on these bundles. Simplification thus enables the mechanism designer to guide equilibrium selection, and our results suggest that the presence or absence of such guidance can have a significant impact on the economic properties of the mechanism. 

The informational assumptions underlying our analysis are crucial, and the amount of information available to the agents plays an important role for the tradeoff between simplicity and expressiveness. 
In the sponsored search setting, both the existence of a zero-revenue equilibrium in the expressive mechanism, and the existence of the desirable equilibrium in the simplification, rely on the assumption of complete information. In other words, agents can on one hand use information about each others' types to coordinate and harm the auctioneer, but on the other hand, the same information guarantees that simplified mechanisms retain the desirable equilibria of the expressive mechanism. 
In combinatorial auctions the contrast is equally stark: while bids on every single package may be required to sustain an efficient equilibrium in the incomplete information setting, we show that such an equilibrium can be obtained with a number of bids that is quadratic in the number of agents, and potentially exponentially smaller than the number of bundles, given that agents have complete information. 

\paragraph{Related Work}
Several authors have criticized the revelation principle because it does not take computational aspects of mechanisms into account. In this context, \citet{CoSa03} consider sequential mechanisms that reduce the amount of communication, and non-truthful mechanisms that shift the computational burden of executing the mechanism, and the potential loss when it is executed suboptimally, from the designer to the agents. 
\citet{HyBo07a,HyBo07b} propose to circumvent computational problems associated with direct type revelation via the automated design of partial-revelation mechanisms, and in particular study approximately incentive compatible 
mechanisms that do not make any assumptions about agents' preferences. This approach is very general, but also hard to analyze theoretically,
with complex, regret-based algorithms.

\citet{BNS07} and \citet{FeBl06} consider settings with one-dimensional types and ask how much welfare and revenue can be achieved by mechanisms with a bounded message space. By contrast, we study mechanisms with message spaces that grow in some parameter of the problem and may even have infinite size, and obtain results both for one-dimensional and multi-dimensional types. 

A different notion of simplicity of a mechanism was considered by \citet{BaRo10a}: the authors show that among all payment rules that guarantee an efficient equilibrium when ranking agents according to their bids, the GSP payment rule is optimally simple in the sense that prices depend on bids in a minimal way. 

\citet{SSOA08a} study the tradeoff between simplicity and revenue in the context of pricing rules for communication networks and define the ``price of simplicity'' as the ratio between the revenue of a very simple pricing rule and the maximum revenue that can be obtained. 



\section{Preliminaries}

A mechanism design problem is given by a set~$N=\{1,2,\dots,n\}$ of \emph{agents} that interact to select an element from a set~$\Omega$ of \emph{outcomes}.  Agent $i\in N$ is associated with a \emph{type}~$\theta_i$ from a set $\Theta_i$ of possible types, representing private information held by this agent.  We write $\theta=(\theta_1,\theta_2,\dots,\theta_n)$ for a profile of types for the different agents, $\Theta=\prod_{i\in N}\Theta_i$ for the set of possible type profiles, and $\theta_{-i}\in\Theta_{-i}$ for a profile of types for all agents but~$i$. 
Each agent $i\in N$ further employs preferences over~$\Omega$, represented by a \emph{valuation function} $v_i:\Omega\times\Theta_i\rightarrow\reals$. 
The quality of an outcome $o\in\Omega$ is typically measured in terms of its social welfare, which is defined as the sum $\sum_{i\in N} v_i(o,\theta_i)$ of agents' valuations. An outcome that maximizes social welfare is also called \emph{efficient}. 

A \emph{mechanism} is given by a tuple $(N,X,f,p)$, where $X=\prod_{i\in N} X_i$ is a set of \emph{message profiles}, $f:X\rightarrow\Omega$ is a \emph{social choice function}, and $p:X\rightarrow\reals^n$ is a \emph{payment function}.  In this paper, we mostly restrict our attention to direct mechanisms, \ie mechanisms where $X_i\subseteq\Theta_i$ for every $i\in N$. A direct mechanism $(N,X,f,p)$ with $X=\Theta$ is called \emph{efficient} if for every $\theta\in\Theta$, $f(\theta)$ is an efficient outcome. 
Just as for type profiles, we write $x_{-i}\in X_{-i}$ for a profile of messages by all agents but~$i$.  
We assume quasilinear preferences, \ie the utility of agent~$i$ given a message profile $x\in X$ is $u_i(x,\theta_i)=v_i(f(x),\theta_i)-p_i(x)$.  
The \emph{revenue} achieved by mechanism $(N,X,f,p)$ for a message profile $x\in X$ is $\sum_{i\in N}p_i(x)$. 

Game-theoretic reasoning is used to analyze how agents interact with a mechanism, a desirable criterion being stability according to some game-theoretic solution concept.  We consider two different settings. In the \emph{complete information} setting, agents are assumed to know the type of every other agent.  A strategy of agent~$i$ in this setting is a function $s_i:\Theta\rightarrow X_i$.  In the \emph{(strict) incomplete information} setting, agents have no information, not even distributional, about the types of the other agents.  A strategy of agent~$i$ in this setting thus becomes a function $s_i:\Theta_i\rightarrow X_i$. 

The two most common solution concepts in the complete information setting are dominant strategy equilibrium and Nash equilibrium.  A strategy $s_i:\Theta\rightarrow X_i$ is a \emph{dominant strategy} if for every $\theta\in\Theta$, every $x_{-i}\in X_{-i}$, and every $x_i\in X_i$, 
\[
u_i((s_i(\theta),x_{-i}),\theta_i)\geq u_i((x_i,x_{-i}),\theta_i).
\]
A profile $s\in\prod_{i\in N}s_i$ of strategies $s_i:\Theta\rightarrow X_i$ is a \emph{Nash equilibrium} if for every $\theta\in\Theta$, every $i\in N$, and every \text{$s_i':\Theta\rightarrow X_i$,} 
\[
u_i((s_i(\theta),s_{-i}(\theta)),\theta_i)\geq u_i((s_i'(\theta),s_{-i}(\theta)), \theta_i).
\]
The existence of a dominant strategy $s_i:\Theta\rightarrow X_i$ always implies the existence of a dominant strategy $s'_i:\Theta_i\rightarrow X_i$ that does not depend on the types of other agents.  The solution concept of dominant strategy equilibrium thus carries over directly to the incomplete information setting. Formally, a strategy $s_i:\Theta_i\rightarrow X_i$ is a \emph{dominant strategy} in the incomplete information setting if for every $\theta_i\in\Theta_i$, every $x_{-i}\in X_{-i}$, and every $x_i\in X_i$, 
\[
u_i((s_i(\theta_i),x_{-i}),\theta_i)\geq u_i((x_i,x_{-i}),\theta_i).
\]
The appropriate variant of the Nash equilibrium concept in that setting is that of an ex-post equilibrium.  A profile $s\in\prod_{i\in N}s_i$ of strategies $s_i:\Theta_i\rightarrow X_i$ is an \emph{ex-post equilibrium} if for every $\theta\in\Theta$, every $i\in N$, and every $s_i':\Theta_i\rightarrow X_i$, 
\[
u_i((s_i(\theta_i),s_{-i}(\theta_{-i})),\theta_i)\geq u_i((s_i'(\theta_i),s_{-i}(\theta_{-i})),\theta_i).
\]

We conclude this section with a direct mechanism due to \citet{Vick61a}, \citet{Clar71a}, and \citet{Grov73a}.  This mechanism starts from an efficient social choice function~$f$ and computes each agent's payment according to the total value of the other agents, thus aligning his interests with that of society.  Formally, mechanism $(N,X,f,p)$ is called Vickrey-Clarke-Groves (VCG) mechanism\footnote{Actually, we consider a specific member of a whole family of VCG mechanisms, namely the one that uses the \citeauthor{Clar71a} pivot rule.} if $X=\Theta$,~$f$ is efficient, and 
\[
p_i(\theta) = \max_{o\in\Omega} \sum_{j\neq i} v_j(o,\theta_j) - \sum_{j\neq i} v_j(f(\theta),\theta_j).
\]
In the VCG mechanism, revealing types $\theta\in\Theta$ truthfully is a dominant strategy equilibrium~\citep{Grov73a}. We will refer to the resulting outcome as the VCG outcome for~$\theta$, and write $R(\theta)$ for the revenue obtained in this outcome. 

\section{Simplifications} 
\label{sec:simplifications}

Our main object of study in this paper are simplifications of a mechanism obtained by restricting its message space.  Consider a mechanism $M=(N,X,f,p)$.  A mechanism $\hM=(N,\hX,\hf,\hp)$ will be called a \emph{simplification} of~$M$ if $\hX\subseteq X$, $\hf|_{\hX}=f|_{\hX}$, and $\hp|_{\hX}=p|_{\hX}$.  

We will naturally be interested in the set of outcomes that can be obtained in equilibrium, both in the original mechanism~$M$ and the simplified mechanism~$\hM$.\footnote{In the following, we will simply talk about equilibria without making a distinction between the different equilibrium notions.  Unless explicitly noted otherwise, our results concern Nash equilibria in the complete information case and ex-post equilibria in the incomplete information case.}  
\citet{Milg10a} defines a property he calls tightness, which requires that the simplification does not introduce any additional equilibria.  More formally,  simplification $\hM$ will be called \emph{tight} if every equilibrium of $\hM$ is an equilibrium of $M$.\footnote{\citet{Milg10a} considers a slightly stronger notion of tightness defined with respect to (pure-strategy) $\epsilon$-Nash equilibria.}  Tightness ensures that the simplified mechanism is at least as good as the original one with respect to the worst outcome obtained in any equilibrium.  It does not by itself protect good equilibrium outcomes, and we will in fact see examples of tight simplifications that eliminate \emph{all} ex-post equilibria. 

A property that will be useful in the following is a variant of \citeauthor{Milg10a}'s outcome closure for exact equilibria. It requires that for every agent, and for every choice of messages from the restricted message sets of the other agents, it is optimal for the agent to choose a message from his restricted message set. 
More formally, a simplification $(N,\hX,\hf,\hp)$ of a mechanism $(N,X,f,p)$ satisfies \emph{outcome closure} if for every $\theta\in\Theta$, every $i\in N$, every $\hx_{-i}\in\hX_{-i}$, and every $x_i\in X_i$ there exists $\hx_{i}\in\hX_i$ such that $u_i((\hx_i,\hx_{-i}),\theta_i) \geq u_i((x_i,\hx_{-i}),\theta_i)$. 

This turns out to be sufficient for tightness in both the complete and incomplete information case.\footnote{\citeauthor{Milg10a} stated the result only for complete information, but the proof goes through for incomplete information as well.} 
\begin{proposition}[\citet{Milg10a}] \label{prop:tightness}
Every simplification that satisfies outcome closure is tight with respect to both Nash and ex-post Nash equilibria. 
\end{proposition}
\begin{proof}
Fix $\theta\in\Theta$.  Consider a mechanism $M=(N,X,f,p)$, a simplification $\hM=(N,\hX,\hf,\hp)$ that satisfies outcome closure, and a Nash equilibrium $\hx$ of $\hM$.  Assume for contradiction that $\hx$ is not a Nash equilibrium of $M$. Then, for some $i\in N$, there exists $x'_i \in X_i$ such that $u_i((x'_i,\hx_{-i}),\theta_i) > u_i(\hx,\theta_i)$.  Since $\hM$ satisfies outcome closure, there further exists $\hx'_i\in\hX_i$ such that $u_i((\hx'_i,\hx_{-i}),\theta_i) \ge u_i((x'_i,\hx_{-i}),\theta_i)$.  It follows that $u_i((\hx'_i,\hx_{-i}),\theta_i) > u_i(\hx,\theta_i)$, which contradicts the assumption that $\hat{x}$ is a Nash equilibrium of $\hM$. 
\end{proof}

One way to guarantee good behavior in the \emph{best case} is by requiring that a simplification $\hM$ preserves all equilibria of the original mechanism $M$, in the sense that for every Nash equilibrium of $M$, there exists an equilibrium of $\hM$ that yields the same outcome and payments. 
We will call a simplification satisfying this property \emph{total}, and will return to total mechanisms in \secref{sec:information}. To the end of equilibrium selection totality clearly needs to be relaxed, by requiring that only certain desirable outcomes are preserved. A typical desirable outcome in many settings is the VCG outcome. For example, although this outcome has some shortcomings in fully general combinatorial auction domains~\cite{AuMi06a}, it remains of significant interest in settings with unit-demand
preferences, such as sponsored search. We will call simplification~$\hM$ \emph{Vickrey-preserving} if for every $\theta\in\Theta$, it has an equilibrium that yields the VCG outcome for~$\theta$.

\section{Sponsored Search Auctions} 
\label{sec:ssa}

In sponsored search~\citep[see, \eg][]{LPSV07a}, the agents compete for elements of a set $S=\{1,\dots,k\}$ of slots, where $k\leq n$.  Each outcome corresponds to a one-to-one assignment of agents to slots, \ie $\Omega\subseteq\{1,\dots,n\}^N$ such that $o_i\neq o_j$ for all $o\in\Omega$ and $i,j\in\{1,\dots,n\}$ with $i\neq j$.  We will assume that $v(o,\theta_i)=0$ if $o_i>k$, and that there are no externalities, \ie $v_i(o,\theta_i)=v_i(o',\theta_i)$ if $o_i=o'_i$.  In slight abuse of notation, we will write $v_i(j,\theta_i)$ for the valuation of agent $i$ for slot $j$. 

We consider simplifications of two mechanisms, the Vickrey-Clarke-Groves (VCG) mechanism and the Generalized Second Price (GSP) mechanism, and analyze their behavior for different spaces of type profiles, which we denote by $\Theta$, $\Theta^>$, and $\Theta^\alpha$. In $\Theta$, valuations can be arbitrary non-negative numbers. 
$\Theta^>$ adds the restriction that valuations are strictly decreasing, \ie $v_i(j,\theta_i)>v_i(j+1,\theta_i)$ for every $\theta\in\Theta^>$, $i\in N$, and $j\in\{1,\dots,k-1\}$.  Valuations in $\Theta^\alpha$ are assumed to arise from ``clicks'' associated with each slot and a valuation per click. In other words, there exists a fixed click-through rate vector $\alpha\in\ctrf=\{\alpha'\in\reals^n \midd \text{$\alpha'_i>\alpha_j$ if $i<j$}\}$, which may or may not be known to the mechanism, and $v_i(j,\theta_i)=\alpha_j\cdot v_i(\theta_i)$ for some $v_i(\theta_i)\in\reals_{\geq 0}$.  Thus, $\alpha_j=0$ if $j>k$, and it will be convenient to assume that $\alpha_1=1$. We finally define $\Psi=\bigcup_{\alpha\in\ctrf}\Theta^\alpha$, and observe that $\Psi\subset\Theta^>\subset\Theta$. 

The message an agent $i\in N$ submits to the mechanisms thus corresponds to a vector of bids $x_{i,j}\in\reals$ for slots $j\in S$. Given a message profile $x\in X$, the VCG mechanism assigns each agent $i$ a slot $f_i(x)=o_i$ so as to maximize $\sum_i x_{i,o_i}$, and charges that agent $p_i(x)=\max_{o' \in \Omega} \sum_{j \neq i} x_{j,o'_j} - \sum_{j \neq i} x_{j,o_j}$. The GSP mechanism is defined via a sequence of second-price auctions for slots $1$ through $k$: slot $j$ is assigned to an agent $i$ with a maximum bid for that slot at a price equal to the second highest bid, both with respect to the set of agents who have not yet been assigned a slot, \ie $f_i(x)=o_i$ such that $x_{i,j} = \max_{i'\in N: o_{i'} \ge j} x_{i',j}$ and $p_i(x)=\max_{i'\in N: o_{i'} > j} x_{i,j}$.

\subsection{Envy-Freeness and Efficiency} 

The original analysis of GSP due to \citet{EOS07a} and \citet{Vari07a} focuses on equilibria that are ``locally envy-free.''\footnote{\citeauthor{Vari07a} calls these equilibria ``symmetric equilibria.''} Assume that $\theta\in\Theta^\alpha$, and consider an outcome in which agent~$i$ is assigned to slot~$i$, for all $i\in N$. Such an outcome is called locally envy-free if, in addition to being a Nash equilibrium, no agent could increase his utility by exchanging bids with the agent assigned the slot directly above him, \ie if for every $i\in\{2,\dots,n\}$, $\alpha_i\cdot v_i - p_i \geq \alpha_{i-1}\cdot v_i - p_{i-1}$. Restricting attention to envy-free equilibria immediately solves all revenue problems: as \citeauthor{EOS07a} point out, revenue in any locally envy-free equilibrium is at least as high as that in the dominant-strategy equilibrium of the VCG auction. Conversely, \citeauthor{Milg10a}'s observation concerning zero-revenue equilibria contains an implicit critique of the assumptions underlying the restriction to equilibria that are envy-free. It will be instructive to make this critique explicit. 

\citeauthor{EOS07a} argue that an equilibrium where some agent~$i$ envies some other agent~$j$ assigned the next higher slot is not a reasonable rest point of the bidding process, because agent~$i$ might increase the price paid by agent~$j$ without the danger of harming his own utility should~$j$ retaliate. There are two problems with this line of reasoning. First, it is not clear why agent~$j$ should retaliate, especially if he is worse off by doing so. Second, agent~$i$ might in fact have a very good reason \emph{not} to increase the price paid by~$j$, like a desire to keep prices low in the long run through tacit collusion. 

We consider a weaker refinement of Nash equilibrium, by asking under what conditions there exists a bid~$b_i$ for agent~$i$ such that (i)~agent~$j$ is forced out of the higher slot, in the sense that it becomes a better response for~$j$ to underbid~$i$, and (ii)~agent~$i$ is strictly better off after this response by $j$ than at present. This is the case exactly when 
\[
	\alpha_i \cdot v_j - p_i > \alpha_j \cdot v_j - b_i 
	\qquad\text{and}\qquad
	\alpha_j\cdot v_i - b_i > \alpha_i\cdot v_i - p_i,
\]
	where $p_i$ is the price currently paid by agent~$i$. The second inequality assumes that $j$ will respond by bidding just below the
bid $b_i$ of agent $i$, such that this becomes $i$'s new price. Rewriting, 
this deviation is possible when
\[
b_i > (\alpha_j-\alpha_i)\cdot v_j + p_i \qquad\text{and}\qquad b_i < (\alpha_j-\alpha_i)\cdot v_i +p_i.
\]
	Clearly, a bid $b_i$ with this property exists if and only if $v_i>v_j$. In turn, agents~$i$ and~$j$ such that $i=j+1$ and $v_i>v_j$ exist if and only if the current assignment is inefficient. 
	
	This provides a very strong argument against inefficient equilibria as rest points of the bidding process, much stronger than the argument against equilibria that are not envy-free. In the context of sponsored search auctions we will therefore restrict our attention to efficient equilibria. It is worth noting at this point that the set of efficient equilibria forms a strict superset of the set of locally envy-free equilibria, and in particular contains the zero-revenue equilibria identified by \citeauthor{Milg10a} and discussed next. 

\subsection{Comments on \citeauthor{Milg10a}'s Analysis}

\citet{Milg10a} observed that for every profile of agent types, both VCG and GSP have a Nash equilibrium that yields zero revenue, and that this equilibrium in fact becomes the unique equilibrium if there is a small cost associated with submitting a positive bid. To alleviate this fact, he proposed to restrict the message space of both VCG and GSP to $\hX=\{(\alpha_1\cdot b_i,\dots,\alpha_k\cdot b_i)\midd b_i\in\reals_{\ge 0}\}$ for some $\alpha\in\ctrf$. 

The following proposition summarizes our knowledge about the resulting simplifications, which we will refer to as \aVCG and \aGSP. Most of these observations were already proved or at least claimed by \citeauthor{Milg10a}, but a proof of the proposition is given in \appref{app:ssa} for the sake of completeness. 
\begin{proposition} \label{prop:ssa}
	Let $\alpha\in\ctrf$.  Then, \aGSP and \aVCG are tight on $\Theta$, have positive revenue on $\Theta^>$ if $n,k\geq 2$, and are Vickrey-preserving on $\Theta^{\alpha}$. 
\end{proposition}

Assuming that the click-through rate vector $\alpha$ is known, both \aVCG and \aGSP look very appealing: they eliminate all zero-revenue equilibria, without affecting the truthful equilibrium and without introducing any new equilibria. 

In practice, however, the relevant click-through rate may not be known. A somewhat more realistic model assumes a certain degree of heterogeneity among the population generating the clicks.  More precisely, a certain fraction of this population is assumed to be ``merely curious,'' such that clicks by this part of the population do not generate any value for the agents.  This introduces an information asymmetry, where the mechanism observes the overall click-through rate vector $\alpha$, while agents derive value from a different click-through rate vector $\beta$.\footnote{As \citeauthor{Milg10a} points out, even the most sophisticated machine learning techniques might not be able to predict the click-through rates from which agents actually derive value.}  In the following we will assume that $\beta$ is the same for all agents, and that $\beta$ is again normalized such that $\beta_1=1$ and $\beta_j=0$ if $j>k$. 

For a slight variation of our model, in which there is a small dependence between $\alpha$ and $\beta$,\footnote{\citeauthor{Milg10a} assumes that there is a fraction $\lambda$ of shoppers with click-through rate vector $\alpha$ and a fraction $(1-\lambda)$ of curious searchers with click-through rate vector $\beta$.  The click-through rate vector $\gamma$ observed by the search provider is then given by $\gamma_j=\lambda \cdot \alpha_j+(1-\lambda)\cdot \beta_j$.} \citeauthor{Milg10a} established a separation between GSP and VCG: \aGSP retains the VCG outcome while \aVCG fails. 
We obtain an analogous observation in our model. First, $\alpha$-GSP is Vickrey-preserving, as we prove in \appref{app:milgrom}. 
\begin{proposition} \label{prop:milgrom}
	Let $\alpha,\beta\in\ctrf$.  Then, \aGSP is Vickrey-preserving on $\Theta^\beta$ if and only if the sequence $\{p_j(\theta)/\alpha_j\}_{j=1,\dots,k}$, where $p_j(\theta)=\sum_{i=j}^k (v_{i+1}(\theta_{i+1})\cdot(\beta_i-\beta_{i+1}))$, is decreasing.
\end{proposition} 

Second, we establish in \appref{app:counterexample} that \aVCG is not Vickrey-preserving. But this line of reasoning seems a bit problematic for two reasons. First, there seems no reason to believe that the condition relating prices (and thus bids) and click-through rates in \propref{prop:milgrom} would be satisfied in practice. Second, the above discussion only shows superiority of GSP over VCG with respect to a \emph{particular} simplification, and it might well be the case that there exists a different simplification of VCG mechanism with comparable or even better properties. 


An additional observation that we make, in regard to the ability of these
simplifications to eliminate zero-revenue equilibria, is that there exist type profiles for which the minimum equilibrium revenue can be \emph{arbitrarily small} compared to the revenue obtained in the VCG outcome.
\begin{theorem} \label{thm:epsilon}
	Let $\epsilon,r>0$.  Then there exist $\alpha\in\ctrf$ and $\theta\in\Theta^\alpha$ such that $R(\theta)\geq r$ and \aVCG has an equilibrium with revenue at most $\epsilon$.  Similarly, there exist $\alpha\in\ctrf$ and $\theta\in\Theta^\alpha$ such that $R(\theta)\geq r$ and \aGSP has an equilibrium with revenue at most $\epsilon$. 
\end{theorem}
\begin{proof}
	We consider a setting with three agents and three slots. The construction can easily be extended to an arbitrary number of agents and slots.
	
	For \aVCG, let $v_i(\theta_i) = r+1$ for all $i\in N$.  Let $\alpha_1=1$, $\alpha_2=1/(r+1)$, and $\alpha_3=1/(2r+2)$.  It is easily verified that the bids $b_1=r+1$ and $b_2=b_3=\epsilon$ form a Nash equilibrium of \aVCG.  Given these bids, \aVCG assigns slot~$1$ to agent~$1$ at price $\epsilon-\epsilon/(2r+2)$, and slots~$2$ and~$3$ to agents~$2$ and~$3$ at prices $\epsilon/(r+1)-\epsilon/(2r+2)$ and zero.  This yields revenue~$\epsilon$.  In the truthful equilibrium of the VCG mechanism, on the other hand, the price is~$r+\epsilon/(r+1)-\epsilon/(2r+2)$ for the first slot, $\epsilon/(r+1)-\epsilon/(2r+2)$ for the second slot, and zero for the third slot, for an overall revenue of $r+\epsilon/(r+2)$. 
	
	For \aGSP, again let $v_i(\theta_i)=r+1$ for all $i\in N$.  Let $\delta=\epsilon/(r+2)$, $\alpha_1=1$, $\alpha_2=(1+\delta)/(r+1)$, and $\alpha_3 = 1/(r+1)$.  It is easily verified that the bids $b_1=r+1$ and $b_2=b_3=\delta/(1+\delta)\cdot(r+1)$ form a Nash equilibrium of \aGSP.  Given these bids, \aGSP assigns slot~$1$ to agent~$1$ at price $\delta/(1+\delta)\cdot (r+1)$, and slots~$2$ and~$3$ to agents~$2$ and~$3$ at prices~$\delta$ and zero.  This yields revenue $\delta/(1+\delta)\cdot(r+1)+\delta\leq\epsilon$.  In the truthful equilibrium of the VCG mechanism, on the other hand, the price is~$r$ for the first slot, $\delta$ for the second slot, and zero for the third slot, for an overall revenue of $r+\delta$. 
\end{proof}

\subsection{A Sense in which GSP is Superior to VCG}

The above observations raise the following prominent question: \emph{does there exist a simplification that preserves the VCG outcome despite ignorance about the true click-through rates that affect bidders' values, and if so, can this simplification achieve improved revenue relative to the VCG outcome, in every equilibrium?} 

For GSP the answer is surprisingly simple: a closer look at the proofs of \propref{prop:ssa} and \propref{prop:milgrom} reveals that by ignoring the observed click-through rates $\alpha$, and setting $\alpha=\one=(1,\dots,1)$ instead, one obtains a simplification that is tight, guarantees positive revenue, and is Vickrey-preserving on \emph{all} of $\Psi$. This strengthens 
Proposition~\ref{prop:ssa} over the claims for \aGSP and \aVCG.
\begin{corollary} \label{cor:onegsp}
	\aGSP is tight on $\Psi$, has positive revenue on $\Psi$ if $n,k \ge 2$, and is Vickrey-preserving on $\Psi$, if and only if $\alpha=\one=(1,\dots,1)$.
\end{corollary}
The direction from left to right follows by observing that, for every
$\alpha\neq \one$, we can find a $\beta$ such that the condition of
\propref{prop:milgrom} is violated.

In light of \thmref{thm:epsilon}, and given the arguments in favor of efficient equilibria, we may further ask for the minimum revenue obtained by \oneGSP in any efficient equilibrium. It turns out that with the exception of the first slot, \oneGSP always recovers at least half of the VCG revenue. 
\begin{theorem} \label{thm:lowerbound}
	Let $\beta\in\ctrf$, $\theta\in\Theta^\beta$.  Then, every efficient equilibrium of \oneGSP for $\theta$ yields revenue at least 
	\[
		\frac{1}{2} \left(R(\theta) - \sum_{j=1}^k (\beta_j-\beta_{j+1}) \cdot v_{j+1}(\theta_{j+1}) \right) \text{.}
	\]
\end{theorem}
\begin{proof}
	Assume that agents are ordered such that $v_1(\theta_1)\geq\dots\geq v_n(\theta_n)$ and consider a bid profile $b(\theta)$ corresponding to an efficient equilibrium of \oneGSP. 
	It then holds that $b_1(\theta)\geq\dots\geq b_n(\theta)$, and for all $i\in\{1,\dots,k\}$, \oneGSP assigns slot~$i$ to agent~$i$ at price $p_i(\theta)=b_{i+1}(\theta)$. 	
	A necessary condition for $b(\theta)$ to be an equilibrium is that for every agent $j\in N$, $b_j(\theta)$ is large enough such that none of the agents $i>j$ would prefer being assigned slot $j$ instead of $i$. In particular, for every $i\in N$,
	\[
	\begin{split}
		\beta_{i+1}\cdot v_{i+1}(\theta_{i+1}) - p_{i+1}(\theta) &\geq \beta_i\cdot v_{i+1}(\theta_{i+1}) - b_i(\theta) \quad \text{and} \\
		\beta_{i+2}\cdot v_{i+2}(\theta_{i+2}) - p_{i+2}(\theta) &\geq \beta_i\cdot v_{i+2}(\theta_{i+2}) - b_i(\theta) \text{.}
	\end{split}
	\]
	Since $p_i(\theta)=b_{i+1}(\theta)$ and by rearranging, 	
	\[
	\begin{split}
		b_i(\theta) &\geq (\beta_i-\beta_{i+1}) \cdot v_{i+1}(\theta_{i+1}) + p_{i+1}(\theta) 
		= (\beta_i-\beta_{i+1}) \cdot v_{i+1}(\theta_{i+1}) + b_{i+2}(\theta) \quad\text{and} \\
		b_i(\theta) &\geq (\beta_i - \beta_{i+2}) \cdot v_{i+2}(\theta_{i+2}) + p_{i+2}(\theta) 
		\geq (\beta_{i+1}-\beta_{i+2}) \cdot v_{i+2}(\theta_{i+2}) + b_{i+3}(\theta) \text{.}  
	\end{split}
	\]
	If we repeatedly substitute according to the first inequality, we obtain
	\[
	\begin{split}
		b_i(\theta) &\geq \sum_{j=1}^{\mathclap{\left\lfloor\frac{k-i}{2}\right\rfloor}} (\beta_{i+2\cdot j-2}-\beta_{i+2\cdot j-1}) \cdot v_{i+2\cdot j-1}(\theta_{i+2\cdot j-1}) \quad \text{and} \\
		b_i(\theta) &\geq \sum_{j=1}^{\mathclap{\left\lfloor\frac{k-i}{2}\right\rfloor}} (\beta_{i+2\cdot j-1}-\beta_{i+2\cdot j}) \cdot v_{i+2\cdot j}(\theta_{i+2\cdot j}) \text{.}
	\end{split}
	\]
	By adding the two inequalities, 
	\[
		2\cdot b_{i}(\theta) \geq \sum_{j=i}^{k} (\beta_j-\beta_{j+1}) \cdot v_{j+1}(\theta_{j+1}) \text{,}
	\]
	and, since $p_i(\theta)=b_{i+1}(\theta)$,
	\[
		2\cdot p_{i}(\theta) \geq \sum_{j=i+1}^{k} (\beta_j-\beta_{j+1}) \cdot v_{j+1}(\theta_{j+1}) \text{.}
	\]
	Now recall that $R(\theta)=\sum_{i\in N}r_i(\theta)$, where
	\[
		r_i(\theta) = \sum_{j=i}^k (\beta_j-\beta_{j+1}) \cdot v_{j+1}(\theta_{j+1}) \text{.}
	\]
	Thus, 
	\[
		\sum_{i\in N} 2\cdot p_i(\theta) \geq \sum_{i\in N} r_i(\theta) - \sum_{j=1}^k (\beta_j-\beta_{j+1}) \cdot v_{j+1}(\theta_{j+1}) \text{.}
	\]
	The revenue obtained in any efficient equilibrium of \oneGSP is therefore at least
	\[
		\sum_{i\in N} p_i(\theta) \geq \frac{1}{2} \left(R(\theta) - \sum_{j=1}^k (\beta_j-\beta_{j+1}) \cdot v_{j+1}(\theta_{j+1}) \right) \text{.} \tag*{\raisebox{-\baselineskip}{\qedhere}}
	\]
\end{proof}

Our analysis also leads to a satisfactory contrast between the
properties of \oneGSP and the properties of \emph{any} simplification of
VCG: any simplification of VCG that does not observe the value-generating click-through rates, and is Vickrey-preserving for all possible choices of these click-through rates, must admit an efficient equilibrium with arbitrarily low revenue. 
\begin{theorem} \label{thm:impossibility}
	Let $\hM$ be a simplification of the VCG mechanism that is Vickrey-preserving on $\Psi$. Then, for every $\theta\in\Psi$ and every $\epsilon>0$, there exists an efficient equilibrium of $\hM$ with revenue at most $\epsilon$. 
\end{theorem}
\begin{proof}
	Fix $\beta\in\ctrf$ and consider an arbitrary type profile $\theta\in\Theta^\beta\subseteq\Psi$. Order the agents such that $v_1(\theta_1) \ge v_2(\theta_2) \ge \dots \ge v_n(\theta_n)$. 
	
	The proof proceeds in two steps. First we will argue that for some $c\geq 0$, every $\delta>0$, and all $i\in N$, $\hX_i$ must contain a message $x_i^{\delta}$ corresponding to bids $b_{ij}$ such that $b_{ii}=\beta_i\cdot v_i(\theta_i)+c$, $b_{ij}\leq\beta_i\cdot v_i(\theta_i)+c+\delta$ for $1\leq j<i$ and $b_{ij}\leq c+\delta$ for $i<j\leq k$.  These messages will then be used to construct an equilibrium with low revenue. 
	
	To show that the restricted message spaces $\hX_i$ must contain messages as described above, we show that these messages are required to reach the VCG outcome for a different type profile $\theta'\in\Theta^{\beta'}\subseteq\Psi$ for a particular $\beta'\in\ctrf$. 
	Denote by $p_i(\theta')$ the price of slot~$i$ for type profile~$\theta'$.  We know that for $j\in\{2,\dots,k\}$, $\beta'_j \cdot v_j(\theta'_j) = \beta'_{j-1} \cdot v_j(\theta'_j) - p_{j-1}(\theta')$~\citep[see, \eg]{Milg10a}.  Consider an arbitrary $\delta>0$.  If we choose $\beta'$ such that $\beta'_1-\beta'_i$ and $\beta'_{i+1}$ are small enough, we can choose $\theta'$ as above such that
	\begin{align*}
		p_j(\theta') - p_i(\theta') & \leq \delta && \text{for $j<i$ and} \\
		p_i(\theta') - p_j(\theta') & \geq \beta_i \cdot v_i(\theta_i) - \delta && \text{for $j>i$.} 
	\end{align*}
	A well-known property of the VCG outcome in the assignment problem is its envy-freeness~\citep[see, \eg]{Leon83a}: denoting by~$b_{ij}$ the bid of agent~$i$ on slot~$j$ and by~$p_j$ the price of slot~$j$, it must hold for every agent~$i$ that
	\begin{align*}
		b_{ii} - p_i \geq b_{ij} - p_j && \text{for all $j\in S$.}
	\end{align*}
	For type profile $\theta'$, we thus obtain
	\begin{align*}
		b_{ij}-b_{ii} &\leq p_j(\theta') - p_i(\theta') \leq \delta && \text{for $j<i$ and} \\
		b_{ii}-b_{ij} &\geq p_i(\theta') - p_j(\theta') \geq \beta_i \cdot v_i(\theta_i) -\delta && \text{for $j>i$.}
	\end{align*}
	
	Using messages $x_i^\delta$, we now construct an efficient equilibrium with low revenue.  Clearly, the allocation that assigns slot~$i$ to agent~$i$ is still efficient under message profile $x^\delta$.  Furthermore, for all $j\in\{1,\dots,k\}$, the VCG price of slot~$j$ under $x^\delta$ goes to zero as~$\delta$ goes to zero.  In particular, we can choose~$\delta$ such that the overall revenue is at most~$\epsilon$. 	We finally claim that there exists some $\delta'>0$ such that~$x^\delta$ is an equilibrium for every~$\delta$ with $\delta'>\delta>0$.  To see this, recall that $\beta_i \cdot v_i(\theta_i)>\beta_j \cdot v_i(\theta_i)$ for $j>i$, so $u_i(x^\delta,\theta_i)>\beta_j \cdot v_i(\theta_i)$ for some small enough~$\delta$. 
	If, on the other hand, agent~$i$ was assigned a slot $j<i$, his payment would be at least $\beta_j \cdot v_j(\theta_j) - \delta > \beta_j \cdot v_i(\theta_i) - \delta$.  This would leave him with utility at most $\delta$, which can be chosen to be smaller than $u_i(x^\delta,\theta_i)$. 
\end{proof}

It is worth noting that despite having a reasonably good lower bound on revenue, \oneGSP does not quite succeed in circumventing \thmref{thm:epsilon}: there exists a type profile for which only the first slot generates a significant amount of VCG revenue, and an equilibrium of \oneGSP for this type profile in which revenue is close to zero.

\section{Combinatorial Auctions} 
\label{sec:ca}

Mechanisms for combinatorial auctions allocate items from a set~$G$ to the agents, \ie $\Omega = \prod_{i\in N} 2^G$ such that for every $o\in\Omega$ and $i,j\in N$ with $i\neq j$, $o_i\cap o_j=\emptyset$. We make the standard assumption that the empty set is valued at zero and that valuations satisfy free disposal, \ie for all $i\in N$, $\theta_i\in\Theta_i$, and $o,o'\in\Omega$,  $v_i(o,\theta_i)=0$ when $o_i=\emptyset$ and $v_i(o,\theta_i)\leq v_i(o',\theta_i)$ when $o_i\subseteq o'_i$. The latter condition also implies that each agent is only interested in the package he receives, and we sometimes abuse notation and write $v_i(C,\theta_i)$ for the valuation of agent~$i$ for any $o\in\Omega$ with $o_i=C$. We further write $k=|G|$ for the number of items, $W(o,x)=\sum_{j\in N}v_j(o,x_j)$ for the social welfare of outcome $o\in\Omega$ under message profile $x\in X$, and $\Wmax(x)=\max_{o\in\Omega}W(o,x)$ for the maximum social welfare of any outcome. Finally, for every agent $i \in N$, message $x_i \in X$, and bundle of items $B \subseteq G$ we write $x_i(B)$ for agent $i$'s bid on bundle $B$. 

The VCG mechanism makes it a dominant strategy for every agent to bid his true valuation for every bundle of items. Since the number of such bundles is exponential in the number of items, however, computational constraints might prevent agents from playing this dominant strategy (even for a well-crafted bidding language~\cite{NiSe06a,LaPa08}). See \appref{app:exponential} for an explicit result on the need to communicate an exponential number of bids in the VCG auction. 
In light of these results, and in light of the observation that simplifications can help to isolate useful equilibria, it is interesting to ask which other (ex-post) equilibria the VCG auction can have. \citet{HM04} showed that these equilibria are precisely the projections of the true types to those subsets of the set of all bundles that form a quasi-field. Let $\Sigma\subseteq 2^G$ be a set of bundles of items such that $\emptyset\in\Sigma$. $\Sigma$ is called a quasi-field if it is closed under complementation and union of disjoint subsets, \ie if 
\begin{itemize}
	\item $B\in\Sigma$ implies $B^c\in\Sigma$, where $B^c=G\setminus B$ and
	\item $B,C\in\Sigma$ and $B\cap C$ implies $B\cup C\in\Sigma$.
\end{itemize}
For a message $x_i\in X_i$, write $x_i^\Sigma$ for the projection of a message $x_i\in X_i$ to $\Sigma$, \ie for the unique message such that for every bundle of items $B \subseteq G$, 
\[
	x_i^\Sigma(B) = \max_{B'\in\Sigma,B'\subseteq B}x_i(B').
\]
The characterization given by \citeauthor{HM04} is subject to the additional constraint of \emph{variable participation}: a strategy profile~$s$ for a set $N$ of agents is an equilibrium of a VCG mechanism under variable participation if for every $N'\subseteq N$, the projection of $s$ to $N'$ is an equilibrium of every VCG mechanism for $N'$. 
\begin{theorem}[\citet{HM04}] \label{thm:hm04}
	Consider a VCG combinatorial auction with a set $N$ of agents and a set $G$ of items. Then, a strategy profile $s=(s_1,\dots,s_n)$ is an ex-post equilibrium of this auction under variable participation if and only if there exists a quasi-field $\Sigma\subseteq 2^G$ such that for every type profile $\theta$ and every agent $i\in N$, $s_i(\theta_i)=\theta_i^\Sigma$. 
\end{theorem}

Intuitively, the social welfare obtained in these ``bundling'' equilibria decreases as the set of bundles becomes smaller. A simple argument shows, for example,  that welfare in the bundling equilibrium for $\Sigma$, where $|\Sigma|=2^m$ for some $m\leq k$, can be smaller by a factor of $k/m$ than the maximum welfare. For this, consider a setting with $k$ agents such that each agent desires exactly one of the items, \ie values this item at~$1$, and each item is desired by exactly one of the agents. Clearly, maximum social welfare is $k$ in this case. On the other hand, since $\Sigma$ is a quasi-field, it cannot contain more than $m$ bundles that are pairwise disjoint. Therefore, by assigning only bundles in $\Sigma$, one can obtain welfare at most~$m$. 

In addition, welfare can also differ tremendously among quasi-fields of equal size, which suggests an opportunity for simplification, 
\begin{proposition} \label{prop:welfare}
Let $G$ be a set of items, $k=|G|$, and $m\leq k$. Then there exist quasi-fields $\Sigma,\Sigma'\subseteq 2^G$ with $|\Sigma|=|\Sigma'|=2^m$ and a type profile $\theta$ such that 
\[
\frac{\Wmax(\theta^{\Sigma})}{\Wmax(\theta^{\Sigma'})} \ge 
\frac{m}{\left\lceil m^2/k \right\rceil} .
\]
\end{proposition}
\begin{proof}
Consider a partition of $G$ into sets $G_1,\dots,G_m$ of size $\lceil \frac{k}{m}\rceil$ or $\lfloor\frac{k}{m}\rfloor$, and let $\Sigma$ be the closure of $\{G_1,\dots,G_m\}$ under complementation and union of disjoint sets. For every $i$, $1\leq i\leq m$, choose an arbitrary $\mathit{g}_i\in G_i$, and define $\theta=(\theta_1,\dots,\theta_m)$ such that for every $i$ with $1\leq i\leq m$, $v_i(X,\theta_i)=1$ if $\mathit{g}_i\in X$ and $v_i(X,\theta_i)=0$ otherwise. Clearly, $\Wmax(\theta^{\Sigma})=m$. Now consider a second partition of $G$ into sets $G'_1,\dots,G'_m$ of size $\lceil\frac{k}{m}\rceil$ or $\lfloor\frac{k}{m}\rfloor$ such that $G'_j\supseteq \{\,\mathit{g}_i \midd (j-1)\lfloor k/m \rfloor + 1 \leq i \leq j \lfloor k/m \rfloor\,\}$, and let $\Sigma'$ be the closure of $\{G_1,\dots,G_m\}$ under complementation and union of disjoint sets. It is then easily verified that $\Wmax(\theta^{\Sigma'})\leq \lceil m/(k/m) \rceil=\lceil m^2/k \rceil$, and the claim follows. 
\end{proof}

Moreover, agents might disagree about the quality of the different bundling equilibria of a given maximum size. In particular, the set of these equilibria might contain several Pareto undominated equilibria, but no dominant strategy equilibrium. From the point of view of equilibrium selection, this is the worst possible scenario. 
\begin{example}
Let $N=\{1,2,3\}$, $G=\{A,B,C\}$, and consider a type profile $\theta = (\theta_1,\theta_2,\theta_3)$ such that
\begin{align*}
v_1(X,\theta_1) &= \begin{cases}
	4 & \text{if $\{A,C\}\subseteq X$,} \\
	3 & \text{if $A\in X$ and $C\notin X$} \\
	1 & \text{otherwise;} 
\end{cases} \displaybreak[0] \\ 
v_2(X,\theta_2) &= \begin{cases}
	3 & \text{if $\{A,C\}\subseteq X$ or $\{B,C\}\subseteq X$ and} \\
	0 & \text{otherwise;}
\end{cases}	\displaybreak[0] \\ 
v_3(X,\theta_3) &= \begin{cases}
	1 & \text{if $B \in X$} \\
	0 & \text{otherwise.} 
\end{cases}
\end{align*}

Clearly, at least four bids are required to express $\theta_1$. Since a quasi-field on $G$ must contain both the empty set and~$G$ itself, there are four quasi-fields of size four or less:
\begin{align*}
	\Sigma^1 &= \{\emptyset,\{A\},\{B,C\},\{A,B,C\}\} &
	\Sigma^2 &= \{\emptyset,\{B\},\{A,C\},\{A,B,C\}\} \\
	\Sigma^3 &= \{\emptyset,\{C\},\{A,B\},\{A,B,C\}\} &
	\Sigma^4 &= \{\emptyset,\{A,B,C\}\}
\end{align*}
For $i\in\{1,2,3,4\}$, write $\theta^i=(\theta_1^{\Sigma^i},\theta_2^{\Sigma^i},\theta_3^{\Sigma^i})$ for the projection of $\theta$ to $\Sigma^i$. 
The following is now easily verified. In the VCG outcome for $\theta^1$, agent~$1$ is assigned $\{A\}$ at price~$0$ for a utility of~$3$, and agent~$2$ is assigned $\{B,C\}$ at price~$1$ for a utility of~$2$. In the VCG outcome for $\theta^2$, agent~$1$ is assigned $\{A,C\}$ at price~$3$ for a utility of~$1$, and agent~$3$ is assigned $\{B\}$ at price~$0$ for a utility of~$1$. Finally, in the VCG outcome for $\theta^3$ and $\theta^4$, agent~$1$ is assigned $\{A,B,C\}$ at price~$3$ for a utility of~$1$. The outcomes for $\theta^1$ and $\theta^2$ are both Pareto undominated. Also observe that social welfare is greater for $\theta^1$, while $\theta^2$ yields higher revenue. 
\end{example}
\medskip

Finally, the projection to a quasi-field can result in an equilibrium with revenue zero, even if revenue in the dominant strategy equilibrium is strictly positive. This is illustrated in the following example. It should be noted that this example, as well as the previous one, can easily be generalized to arbitrary numbers of agents and items and a large range of upper bounds on the size of the quasi-field. 
\begin{example}  \label{ex:revenue}
Let $N=\{1,2, 3\}$, $G=\{A,B,C,D\}$, and consider a type profile $\theta = (\theta_1,\theta_2,\theta_3)$ such that 
\begin{align*}
	v_1(X,\theta_1) &= \begin{cases}
		2 & \text{if $\{A,D\}\subseteq X$ and} \\
		0 & \text{otherwise;}
	\end{cases} \displaybreak[0] \\ 
	v_2(X,\theta_2) &= \begin{cases}
		2 & \text{if $\{A,B\}\subseteq X$,} \\
		1 & \text{if $B\in X$ and $A\notin X$, and} \\
  	0 & \text{otherwise;}
	\end{cases} \displaybreak[0] \\
	v_3(X,\theta_3) &= \begin{cases}
		2 & \text{if $C\in X$, and} \\
		0 & \text{otherwise.}
	\end{cases}
\end{align*}

In the VCG outcome for $\theta$, agent~$1$ is assigned $\{A,D\}$ at price~$1$, agent~$2$ is assigned $\{B\}$ at price~$0$, and agent~$3$ is assigned~$\{C\}$ at price~$0$, for an overall revenue of~$1$. In the VCG outcome for $\theta^\Sigma$, on the other hand, where $\Sigma=\{\emptyset,AB,CD,ABCD\}$, agent~$2$ is assigned $\{A,B\}$ and agent~$3$ is assigned $\{C,D\}$, both at price~$0$. Revenue is $0$ as well. 
\end{example}
\medskip

In theory, one way to solve these problems is to simplify the mechanism, and artificially restrict the set of bundles agents can bid on. Given a set $\Sigma\subseteq 2^G$ of bundles, call \SigmaVCG the simplification of the VCG mechanism obtained by restricting the message spaces to $\hX_i\subseteq X_i$ such that for every $\hx_i\in\hX_i$ and every bundle of items $B \subseteq G$, 
\[
	\hx_i(B)=\max_{B'\in\Sigma,B'\subseteq B}\hx_i(B').
\] 
In other words, \SigmaVCG allows agents to bid only on elements of $\Sigma$ and derives bids for the other bundles as the maximum bid on a contained bundle. 
It is easy to see that \SigmaVCG is \emph{maximal in range}~\cite{NiRo07}, \ie that it maximizes social welfare over a subset of $\Omega$. It follows that for each agent, truthful projection onto $\Sigma$ is a dominant strategy in \SigmaVCG. 

This shows that simplification can focus attention on a focal (truthful) equilibrium and thus avoid equilibrium selection along a Pareto frontier. As \propref{prop:welfare} and the above examples suggest, this can have a significant positive impact on both social welfare and revenue. 
It does not tell us, of course, how $\Sigma$ should be chosen in practice. One understood approach for maximizing social welfare without any knowledge about the quality of different outcomes, and without consideration to $\Sigma$ being a quasi-field, is to partition $G$ arbitrarily into $m$ sets of roughly equal size, where $m$ is the largest number of bundles agents can bid on. The welfare thus obtained is smaller than the maximum social welfare by a factor of at most $k/\sqrt{m}$~\citep{HKMT04}. If additional knowledge is available, however, it may be possible to improve the result substantially, as the above results comparing the outcomes for different values of $\Sigma$ show. 

With that being said, there are (at least) two additional properties
that are desirable for a simplification: that \SigmaVCG is tight, and
that $\Sigma$ is a quasi-field. Tightness ensures that no additional
equilibria are introduced as compared to the fully expressive VCG
mechanism, such that the quality of the worst equilibrium outcome of
the simplification is no worse than that of the original
mechanism. This can remain important when $\Sigma$ may still be too
large for agents to use its full projection, in which case agents
would again have to select from a large set of possible ex-post
equilibria. By requiring that $\Sigma$ is a quasi-field, in addition to being a dominant strategy equilibrium of \SigmaVCG, truthful projection to $\Sigma$ is also an
ex-post equilibrium of the fully expressive VCG mechanism, and thus
stable against unrestricted unilateral deviations. This ensures that agents do not experience regret,
in the sense 
of being prevented from sending a message they would want to send
given the messages sent by the other agents.

It turns out that these two requirements are actually equivalent, \ie that \SigmaVCG is tight if and only if $\Sigma$ is a quasi-field. The following result holds with respect to both Nash equilibria and ex-post equilibria. 
\begin{theorem}  \label{thm:tight}
	Let $\Sigma\subseteq 2^G$ such that $\emptyset\in\Sigma$. Then, \SigmaVCG is a tight simplification if $\Sigma$ is a quasi-field, and this condition is also necessary if $n\geq 3$. 
\end{theorem}
\begin{proof}
	For the direction from right to left, assume that $\Sigma$ is a quasi-field.  By \propref{prop:tightness} it suffices to show that \SigmaVCG satisfies outcome closure. Fix valuation functions $v_j$ and types $\theta_j$ for all $j\in N$, and consider an arbitrary agent $i\in N$.  We claim that for every $x_i\in X_i$ and every $\hx_{-i}\in\hX_{-i}$, 
\[
	u_i((\theta_i^\Sigma,\hx_{-i}),\theta_i)\geq u_i((x_i,\hx_{-i}),\theta_i).
\]
To see this, observe that there exists $\hat{\theta}_{-i}\in\Theta_{-i}$ such that $\hat{\theta}_{-i}^\Sigma=\hx_{-i}$, and consider the type profile $(\theta_i,\hat{\theta}_{-i})$. \citet{HKMT04} have shown that the projection of the true types to a quasi-field $\Sigma$ is an ex-post equilibrium of the (fully expressive) VCG mechanism. Thus, in particular, $\theta_i^\Sigma$ is a best response to $\hat{\theta}_{-i}^\Sigma=\hx_{-i}$, which proves the claim. 

For the direction from left to right, assume that $\Sigma$ is not a quasi-field. \citet{HKMT04} have shown that in this case, $\theta^\Sigma$ is not an ex-post equilibrium of the VCG mechanism. On the other hand, $\theta^\Sigma$ is a dominant-strategy equilibrium of \SigmaVCG, because \SigmaVCG is maximal in range. This shows that \SigmaVCG is not tight. 
\end{proof}

Together with \thmref{thm:hm04}, this yields a characterization of the ex-post equilibria of \SigmaVCG for the case when $\Sigma$ is a quasi-field. 
\begin{corollary}\label{cor:sigma}
	Let $\Sigma$ be a quasi-field.  Then, $\hx$ is an ex-post equilibrium of \SigmaVCG under variable participation if and only if $\hx=\theta^{\Sigma'}$ for some quasi-field $\Sigma' \subseteq \Sigma$. 
\end{corollary}

Since a quasi-field of size $m$ can contain at most $\log m$ bundles
that are pairwise disjoint, insisting that a simplification be tight
does come at a cost, decreasing the worst-case social welfare in the
truthful projection by an additional factor of up to $\sqrt{m/\log
m}$. Still, as discussed above, tightness brings other advantages to
the simplified mechanism.

\section{The Role of Information}  \label{sec:information}

Like most of the literature on sponsored search auctions, we have assumed that agents have complete information about each others' valuations for the different slots. It turns out that this assumption is crucial, and that the positive results for \aGSP do not extend to incomplete information settings. 

In particular, \aGSP admits an ex-post equilibrium only in very degenerate cases! This complements a result of \citet{GoSw09a}, who showed that \aGSP has an efficient Bayes-Nash equilibrium 
 on type space $\Theta^\alpha$ if and only if $\alpha$ decreases sufficiently quickly. Of course, ex-post Nash equilibrium is stronger than Bayes-Nash equilibrium, but on the other hand our result precludes the existence of the former for a significantly larger type space. The proof of the following theorem is given in \appref{app:gsp_incomplete}. 
\begin{theorem}\label{thm:gsp_incomplete}
	Let $\alpha,\beta\in\ctrf$. Then, \aGSP has an efficient ex-post equilibrium on type space $\Theta^\beta$ if and only if $n\leq 2$ or $k\leq 1$, even if $\alpha=\beta$. 
\end{theorem}

The \aVCG mechanism does only slightly better: it has an ex-post equilibrium only when it allows agents to bid truthfully (it is thus not a meaningful simplification, in that the language is in this sense exact). The proof of this result is given in \appref{app:vcg_incomplete}. 
\begin{theorem}\label{thm:vcg_incomplete}
	Let $\alpha,\beta\in\ctrf$. Then, \aVCG has an efficient ex-post equilibrium on type space $\Theta^\beta$ if and only if $\alpha=\beta$, or $n\leq 2$, or $k\leq 1$. 
\end{theorem}

These results indicate that simplification is not very useful in sponsored search auctions without the assumption of complete information amongst bidders. Interestingly, simplification is also not necessary in this case to preclude equilibria with bad properties, at least in the special case of our model where the valuations are proportional to some (possibly unknown) vector of value-generating click-through rates. In particular, payments in \emph{every} efficient ex-post equilibrium of the fully expressive VCG mechanism equal the VCG payments.  Moreover, truthful reporting is the only efficient ex-post equilibrium when the number of agents is greater than the number of slots. The proof of this result is given in \appref{app:expostofvcg}. 
\begin{proposition}\label{prop:expostofvcg}
	Consider a VCG sponsored search auction with type profile $\theta\in\Theta^\alpha$, and assume that $s=(s_1,\dots,s_n)$ is an efficient ex-post equilibrium. Then, for all $i$ with $1\leq i\leq k$, the payment for slot $i$ in the outcome for strategy profile $s$ equals $p_i=\sum_{j=i}^{\min(k,n-1)}((\alpha_j-\alpha_{j+1})\cdot v_{j+1}(\theta_{j+1}))$. 
	Moreover, if $n>k$, then $s_{i}(k,\theta_i)=\alpha_{k}\cdot v_i(\theta_i)$ for all $i$ with $1\leq i\leq n$.
\end{proposition}

In considering simplifications for combinatorial auctions, we have adopted the standard approach to assume incomplete information amongst agents, and in particular discussed a family of simplifications of the VCG mechanism that offers a tradeoff between social welfare and the amount of information agents have to communicate. 

One may wonder why this tradeoff is necessary, and in how far it depends on the amount of information agents have about each others' types. It turns out that in the complete information case a much smaller number of bids is enough to preserve an efficient equilibrium, and in fact \emph{all} equilibria, of the fully expressive VCG mechanism. 

We show this using a simplification of the VCG mechanisms that we call \nVCG. The message space of \nVCG consists of all bid vectors with at most $n$ non-zero entries, where $n$ is the number of agents. This reduces the number of bids elicited from each agent from $2^k$ to at most $n$, which can be exponentially smaller. Surprisingly, this simplification is both tight and total, \ie the set of Nash equilibria is completely unaffected by this restriction of the message space. The proof of this result is given in \appref{app:nVCG}.
\begin{theorem}\label{thm:nVCG}
The $n$-VCG mechanism is tight and total with respect to Nash equilibria. 
\end{theorem}

\section{Conclusion and Future Work}

In this paper we have studied simplifications of mechanisms obtained by restricting their message space, and have found that they can be used to solve different kinds of equilibrium selection problems that occur in practice. Direct revelation mechanisms typically have several equilibria, which might be more or less desirable from the point of view of the designer. Computational constraints might also imply that only a subset of the equilibria of a mechanisms can actually be achieved, in which case agents might have to select among several Pareto optimal equilibria. 
On the other hand, restricting the message space of a mechanism often reduces the amount of social welfare that can be achieved theoretically, and this seems to pivot on whether or not agents have complete information about each others' types. The choice between mechanisms with different degrees of expressiveness therefore involves a tradeoff between a benefit of simplicity and a price of simplicity. 

The \emph{price} of simplicity can easily be quantified, for example, by the loss in social welfare potentially incurred by a simplification, and has been studied in the context of both sponsored search and combinatorial auctions. 
\citet{AGV07a} and \citet{BHN08a}, among others, have given bounds on the loss of social welfare incurred by \aGSP for different classes of valuations that are not proportional to the click-through rates. 
\citet{CKS08a} and \citet{BhRo11a} have studied the potential loss in welfare when a set of items is sold through a simplification of the combinatorial auction in which only bids that are additive in items are allowed. This bidding language requires only a small number of bids, and welfare in equilibrium turns out to be smaller by at most a logarithmic factor in the number of agents than the optimum achieved by the fully expressive mechanism. 

The \emph{benefit} of simplicity is much harder to grasp as a concept. In this paper we have argued that simplification can improve the economic properties of a mechanism by precluding bad or promoting good equilibria. A remaining challenge is to understand the benefit of simplicity in the context of simplified mechanisms for which the computation of an equilibrium might be an intractable problem, like the ones of \citeauthor{CKS08a} and \citeauthor{BhRo11a} described above. In contrast to the simplifications considered in the present paper, these mechanisms may not be able to solve the computational or informational problem of enabling agents to bid in a straightforward way. More generally, it is far from obvious how ``straightforwardness" of a mechanism should be measured, but it seems reasonable to require that agents' strategies can be computed in polynomial time.

\bibliography{abb,simplicity}

\appendix

\section{Proof of \propref{prop:ssa}} \label{app:ssa}

%
First we prove that \aGSP and \aVCG are tight on $\Theta.$  By \propref{prop:tightness} it suffices to show that for every $\theta \in \Theta$, every $i \in N$, every $\hx_{-i} \in \hX_{-i}$ and every $x_i \in X_i$ there exists $\hx_i \in \hX_i$ such that $u_i((\hx_i,\hx_{-i}),\theta_i) \ge u_i((x_i,\hx_{-i}),\theta)$.  Denote the outcome of GSP resp.~VCG on $(x_i,\hx_{-i})$ by $o.$  For $l \neq i$ let $b_l \in \reals_{\ge 0}$ be such that $\hx_{l} = (\alpha_1 \cdot b_{l}, \dots, \alpha_k \cdot b_{l}).$  Let $b_{i}$ be the $o_i$-th highest value among the $b_{l}$'s.  Then in the outcome $o'$ of \aGSP and \aVCG on $(\hx_i, \hx_{-i})$ we have $o'_i = o_i$ and, thus, $v_i(o'_i,\theta_i) = v_i(o_i,\theta_i).$  Since $i$'s payment $p_i(o_i)$ and $p_i(o'_i)$ in GSP resp.~VCG and \aGSP resp.~\aVCG only depends on $\hx_{-i}$ and is therefore the same, we conclude that $u_i((\hx_i,\hx_{-i}),\theta_i) \ge u_i((x_i,\hx_{-i}),\theta).$ 

Next we show that \aGSP and \aVCG have positive revenue on $\Theta^>$ for $n,k \ge 2$.  Suppose $\hx$ is a NE of \aGSP resp.~\aVCG.  For $i \in N$ let $b_i \in \reals_{\ge 0}$ be such that $\hx_i = (\alpha_1 \cdot b_i, \dots, \alpha_k \cdot b_i).$  Renumber the agents by non-increasing $b_i$.  The revenue achieved by \aGSP resp.~\aVCG is $\sum_{i} \alpha_i \cdot b_{i+1}$ resp.~$\sum_{i} \sum_{j > i} (\alpha_{j-1} - \alpha_j) \cdot b_j$, where the sums are over all $i, j \le \min(n,k)$ with $b_i, b_j > 0$.  Hence in both, \aGSP and \aVCG, we can only have zero revenue if $b_i = 0$ for all $i > 1$.  In this case all agents but the first remain unassigned and any agent $i > 1$ can bid $0 < b'_i < b_1$ to be assigned slot $2$ at price $0.$  Since $\theta_i \in \Theta^>_i$ we have $v_i(2, \theta_i) > 0$ and, thus, agent $i$'s utility would strictly increase.  We conclude that in both, \aGSP and \aVCG, the revenue associated with $\hx$ must be strictly positive.

Finally, we show that \aGSP and \aVCG are Vickrey compatible on $\Theta^{\alpha}.$  For this denote the outcome and payments computed by VCG for the true types $\theta \in \Theta^{\alpha}$ by $o$ and $p$. Recall that $i < j$ implies that $p_i > p_j.$ We have to argue that there are NE $\hx \in \hX$ of \aGSP and \aVCG in which the outcome and payments are identical to $o$ and $p.$ 

For \aGSP we construct $\hx \in \hX$ as follows:  1.~Renumber the agents by the slot they are assigned to in $o$.  2.~Set $b_1 = v_1(1,\theta_1)$ and for $i > 1$ set $b_i = \alpha_{i-1}^{-1} \cdot p(i-1).$  3.~For all $i$ let $\hx_i = (\alpha_1 \cdot b_i, \dots, \alpha_k \cdot b_i)$.  We have chosen $\hx \in \hX$ such that the outcome and payments computed by \aGSP on $\hx$ are identical to $o$ and $p$.  To see that $\hx$ is a NE of \aGSP observe that if agent $i$ deviates from $\hx$ to win slot $j$, then the price $p''(j)$ that he would have to pay for slot $j$ is at least as large as the price $p'(j) = p(j)$ of slot $j$ in the outcome computed by \aGSP resp.~VCG on $\hx$ resp.~$\theta.$  That is, if $i$ would strictly benefit from the deviation that gives him slot $j$, then we would have $v_i(j, \theta_i) - p(j) \ge v_i(j, \theta_i) - p''(j) > v_i(i,\theta_i) - p'(i) = v_i(i,\theta_i) - p(i)$.  But this would contradict the fact that the outcome and payments computed by VCG are envy free~\citep{Leon83a}. 

For \aVCG we can use $\hx = \theta \in \hX.$  Applying \aVCG to $\hx$ gives, of course, the same assignment and payments as applying VCG to $\theta$.  To see that $\hx$ is a NE of \aVCG observe that any beneficial deviation from $\hx$ in \aVCG would also be a beneficial deviation from $\theta$ in VCG and, thus, would contradict the fact that truthtelling is a dominant strategy of VCG. 

\section{Proof of \propref{prop:milgrom}}
\label{app:milgrom}


 	Fix $\alpha, \beta \in \ctrf$ and $\theta \in \Theta^\beta.$  Renumber the agents so that $v_1(\theta_1) \ge v_2(\theta_2) \ge \dots \ge v_n(\theta_n).$  Efficiency requires that agent $i$ win position $i$.  To get the Vickrey prices $p_i(\theta)$ for all slots $i$ it is necessary that the equilibrium bid by agent $i \in \{2,\dots,\min(n,k+1)\}$ be $b_i(\theta) = p_{i-1}(\theta)/\alpha_{i-1}$. Take $b_1(\theta) = \alpha_1 \cdot v_1(\theta_1)$ and $b_i(\theta) = 0$ for $i > \min(n,k+1).$

 	First suppose that the sequence $\{p_j/\alpha_j\}_j$ is not decreasing, \eg because $p_i(\theta)/\alpha_i < p_{i+1}(\theta)/\alpha_{i+1}.$ It follows that $b_{i+1}(\theta) = p_i(\theta)/\alpha_i < p_{i+1}(\theta)/\alpha_i+1 = b_{i+2}(\theta)$ and, thus, the bids are not ranked in the order required for efficiency.
 
 	Next suppose that the sequence $\{p_j/\alpha_j\}_j$ is decreasing.  In this case the bids are ranked in the correct order.  For a contradiction suppose that some agent $j$ could strictly benefit from a deviation to $\theta' = (\theta_1, \dots, \theta_{j-1}, \theta'_j, \theta_{j+1}, \dots, \theta_n).$ Suppose that given $\theta'$ agent $j$ is assigned slot $l$ at price $p_l(\theta')$. Since agent $j$ strictly benefits from the deviation we must have that $\beta_l \cdot v_j(\theta_j) - p_l(\theta') > \beta_j \cdot v_j(\theta) - p_j(\theta)$.  If $l < j$, then the price that agent $j$ faces for slot $l$ is at least $p_l(\theta') \ge \alpha_l \cdot p_{l-1}(\theta)/\alpha_{l-1} > \alpha_l \cdot p_l(\theta)/\alpha_l = p_l(\theta).$  If $l > j$, then the price that agent $j$ faces for slot $l$ is exactly $p_l(\theta') = p_l(\theta)$.  We conclude that $\beta_l \cdot v_j(\theta_j) - p_l(\theta) \ge \beta_l \cdot v_j(\theta_j) - p_l(\theta') > \beta_j \cdot v_j(\theta_j) - p_j(\theta)$. But this contradicts the envy freeness of the VCG assignment and payments (see, e.g., \cite{Leon83a}). 

\section{Counterexample for the VCG Mechanism} 
\label{app:counterexample}

Let $\alpha\in\ctrf$ be such that $\alpha_1=1$, $\alpha_2=0.5$, and $\alpha_3=0.4$, and let $\beta\in\ctrf$ be such that $\beta_1=1$, $\beta_2=0.9$, and $\beta_3=0.8$.  Let $\theta\in\Theta^\beta$ be such that $v_1(\theta_1)=30$, $v_2(\theta_2)=20$, $v_3(\theta_3)=10$, and $v_i(\theta_i)=0$ if $i>3$.  In the VCG outcome, slot~$1$ is assigned to agent~$1$ at price $p_1(\theta)=(\beta_1-\beta_2)\cdot v_2(\theta_2)+(\beta_2-\beta_3)\cdot v_3(\theta_3)=3$, slot~$2$ to agent~$2$ at $p_2(\theta)=(\beta_2-\beta_3)\cdot v_3(\theta_3)=1$, and slot~$3$ to agent~$3$ at $p_3(\theta)=0$.  
Now assume that the same outcome is obtained in \aVCG, and denote by $b\in\reals^n$ a bid profile that leads to this outcome.  Since both VCG and \aVCG are efficient, it must hold that $b_1\ge b_2\ge b_3$.  To get the same prices as in the VCG outcome, we must further have that $b_2=(\beta_1-\beta_2)/(\alpha_1-\alpha_2)\cdot v_2(\theta_2)=4$ and $b_3=(\beta_2-\beta_3)/(\alpha_2-\alpha_3)\cdot v_3(\theta_3)=10$.  Thus, $b_2<b_3$, which gives a contradiction.

\section{Proof of Theorem~\ref{thm:gsp_incomplete}}
\label{app:gsp_incomplete}

	If $n=1$ or $k=1$, strategies $s_i$ with $s_i(\theta_i)=\beta_1/\alpha_1 \cdot v_i(\theta_i)$ ensure that the agent with the highest valuation is assigned the first slot at a price equal to the second-highest valuation. It is easy to see that this is efficient and constitutes an ex-post equilibrium.  
	
	Now consider the case where $n=2$ and $k\geq 2$.  We claim that $s_i(\theta_i)=(\beta_1-\beta_2)/\alpha_1\cdot v_i(\theta_i)$ for $i\in\{1,2\}$ is the unique efficient ex-post equilibrium in this case. 
	To see this, observe that $s_1$ is an equilibrium strategy if and only if for all $\theta_1\in\Theta_1$ and $\theta_2\in\Theta_2$, 
\begin{align*}
	\beta_1\cdot v_2(\theta_2) - \alpha_1\cdot s_1(\theta_1) &\geq \beta_2\cdot v_2(\theta_2) \quad \text{if agent~$2$ is assigned the first slot, and} \\
	\beta_1\cdot v_2(\theta_2) - \alpha_1\cdot s_1(\theta_1) &\leq \beta_2\cdot v_2(\theta_2) \quad \text{if agent~$2$ is assigned the second slot.}
\end{align*}
	For $s_1$ to be part of an efficient equilibrium, the first inequality has to hold if $v_2(\theta_2)=v_1(\theta_1)+\epsilon$ for any $\epsilon>0$, and the second inequality has to hold if $v_2(\theta_2)=v_1(\theta_1)-\epsilon$ for any $\epsilon>0$.  By rearranging, we get that for every $\theta_1\in\Theta_1$ and every $\epsilon>0$,
\begin{align*}
	s_1(\theta_1) \leq \frac{\beta_1-\beta_2}{\alpha_1}\cdot (v_1(\theta_1)+\epsilon) 
\qquad \text{and} \qquad
	s_1(\theta_1) \geq \frac{\beta_1-\beta_2}{\alpha_1}\cdot (v_1(\theta_1)-\epsilon) \text{.}
\end{align*}
Since analogous conditions have to hold for $s_2$, the claim follows.

Finally consider the case where $n\geq 3$ and $k\geq 2$, and assume for contradiction that~$s$ is an efficient ex-post equilibrium. Observe that~$s$ must remain an equilibrium if we restrict the types in such a way that $v_\ell(\theta_\ell)>0$ if $\ell\in\{1,2\}$ and $v_\ell(\theta_\ell)=0$ otherwise. Strategies $s_1$ and $s_2$ thus have to be of the form described above, \ie $s_i(\theta_i) = (\beta_1-\beta_2)/\alpha_1\cdot v_i(\theta_i)$ for $i\in\{1,2\}$. Similarly,~$s$ remains an equilibrium if we restrict the valuations such that $v_\ell(\theta_\ell)>0$ if $\ell=3$ and $v_\ell(\theta_\ell)=0$ otherwise. It follows that $s_3(\theta_3)>0$ if $v_3(\theta_3)> 0$. Let $v_1(\theta_1)=v_2(\theta_2)=v$ for some $v>0$, choose~$\theta_3$ such that $0 < v_3(\theta_3) < (\beta_1-\beta_2)/\alpha_1 \cdot v$, and let $v_\ell(\theta_\ell)=0$ for $\ell>3$. Thus $s_1(\theta_1) = (\beta_1-\beta_2)/\alpha_1\cdot v$, $s_2(\theta_2) = (\beta_1-\beta_2)/\alpha_1\cdot v$, and $s_\ell(\theta_\ell)=0$ for $\ell>3$.  It further holds that $v_3(\theta_3)<v$, because $\beta_1>\beta_2$ and $\alpha_1=1$. We may now assume without loss of generality that for this bid profile, \aGSP assigns slot~$1$ to agent~$1$, slot~$2$ to agent~$2$, and slot~$3$ to agent~$3$; the case where slot~$1$ is assigned to agent~$2$ and slot~$2$ to agent~$1$ is symmetric.  Agent $2$ thus obtains utility $u_2=\beta_2\cdot v-s_3(\theta_3) < \beta_2\cdot v$.  If he changed his bid to $b_2'>b_2$, he would be assigned slot~$1$ at price $p_1' = \alpha_1\cdot (\beta_1-\beta_2)/\alpha_1\cdot v = (\beta_1-\beta_2)\cdot v$, and obtain utility $u'_2=\beta_1\cdot v-p_1' = \beta_1\cdot v-(\beta_1-\beta_2)\cdot v = \beta_2\cdot v > u_2$.  This contradicts the assumption that~$s$ is an equilibrium. 

\section{Proof of Theorem~\ref{thm:vcg_incomplete}}
\label{app:vcg_incomplete}

	If $\alpha=\beta$, agent~$i$ can bid truthfully on all slots by letting $s_i(\theta_i)=v_i(\theta_i)$.  If $n=1$ or $k=1$, only one slot will be assigned, and agent~$i$ can bid truthfully on this slot by letting $s_i(\theta_i)=\beta_1/\alpha_1\cdot v_i(\theta_i)$. In both cases, truthful bidding constitutes an efficient dominant strategy equilibrium and, a fortiori, an efficient ex-post equilibrium. 
	
	Now consider the case where $n=2$ and $k\geq 2$.  We claim that strategy profile~$s$ where $s_i(\theta_i)=(\beta_1-\beta_2)/(\alpha_1-\alpha_2)\cdot v_i(\theta_i)$ for $i\in\{1,2\}$ is the unique efficient ex-post equilibrium in this case.  To see this, observe that $s_1$ is an equilibrium strategy if and only if for all $\theta_1\in\Theta_1$ and $\theta_2\in\Theta_2$, 
\begin{align*}
	\beta_1\cdot v_2(\theta_2) - (\alpha_1-\alpha_2)\cdot s_1(\theta_1) &\geq \beta_2\cdot v_2(\theta_2) \quad
	\text{if agent~$2$ is assigned the first slot, and} \\
	\beta_1\cdot v_2(\theta_2) - (\alpha_1-\alpha_2)\cdot s_1(\theta_1) &\leq \beta_2\cdot v_2(\theta_2) \quad
	\text{if agent~$2$ is assigned the second slot.}
\end{align*}
	For $s_1$ to be part of an efficient equilibrium, the first inequality has to hold if $v_2(\theta_2)=v_1(\theta_1)+\epsilon$ for any $\epsilon>0$, and the second inequality has to hold if $v_2(\theta_2)=v_1(\theta_1)-\epsilon$ for any $\epsilon>0$.  By rearranging, we get that for every $\theta_1\in\Theta_1$ and every $\epsilon>0$,
\begin{align*}
	s_1(\theta_1) &\leq \frac{\beta_1-\beta_2}{\alpha_1-\alpha_2}\cdot (v_1(\theta_1)+\epsilon) 
	\qquad \text{and} \qquad
	s_1(\theta_1) \geq \frac{\beta_1-\beta_2}{\alpha_1-\alpha_2}\cdot (v_1(\theta_1)-\epsilon) \text{.}
\end{align*}
Since analogous conditions have to hold for $s_2$, the claim follows. 
	
	Finally consider the case where $n\geq 3$, $k\geq 2$, and $\alpha\neq\beta$, and assume for contradiction that~$s$ is an efficient ex-post equilibrium.  We first claim that~$s$ is symmetric. For a contradiction assume that there exist agents~$i$ and~$j$ such that $s_i\neq s_j$.  Observe that~$s$ must remain an equilibrium if we restrict the valuations of all other agents such that in every efficient assignment, agents~$i$ and~$j$ are assigned the last two slots.  This means, however, that $s_i$ and $s_j$ induce an efficient ex-post equilibrium for the case where $n=2$ and $k\geq 2$. By assumption, this equilibrium is asymmetric, contradicting an observation we have made in the previous paragraph. 

	Since $s$ is symmetric, there must exist a function $s_*:\reals\rightarrow\reals$ such that for every $i\in N$ and every $\theta_i\in\Theta_i$, $s_i(\theta_i)=s_*(v_i(\theta_i))$. It is not hard to see that $s_*(v)>0$ if $v>0$ and $s_*(v)=0$ if $v=0$.  Since by convention $\alpha_1=\beta_1=1$ and $\alpha_{k+1}=\beta_{k+1}=0$, and since $\alpha_j\neq\beta_j$ for some $j<k$, there must further exist $i>1$ such that (i)~$\alpha_{j}-\alpha_{j+1}\leq\beta_j-\beta_{j+1}$ for all $j<i$ and $\alpha_i-\alpha_{i+1}>\beta_i-\beta_{i+1}$, or (ii)~$\alpha_{j}-\alpha_{j+1}\geq\beta_j-\beta_{j+1}$ for all $j<i$ and $\alpha_i-\alpha_{i+1}<\beta_i-\beta_{i+1}$. Since the two cases are symmetric, it suffices to consider the first one. 

	Consider an arbitrary $v>0$. We distinguish two cases. 

	First assume that $s_*(v) > (\beta_i-\beta_{i+1})/(\alpha_i-\alpha_{i+1}) \cdot v$.  Let $v_\ell(\theta_\ell)=v$ for $\ell\leq i+1$ and $v_\ell(\theta_\ell)=0$ for $\ell>i+1$.  It then holds that $s_\ell(\theta_\ell) = s_*(v) > 0$ for $\ell\leq i+1$ and $s_\ell(\theta_\ell) = s_*(0) = 0$ for $\ell>i+1$.  Assume without loss of generality that for this bid profile, \aVCG assigns slot $\ell$ to agent $\ell$ for $\ell\leq i+1$, and that the remaining agents are not assigned a slot.  The utility of agent~$i$ in this case is $u_i = \beta_i\cdot v - p_i(\theta) = \beta_i\cdot v - (\alpha_i-\alpha_{i+1})\cdot s_*(v)$.  If he changed his bid to~$b$ with $0<b<s_*(v)$, he would be assigned slot $i+1$ at price~$0$ and obtain utility $\beta_{i+1} \cdot v = \beta_{i+1}\cdot v+(\alpha_{i}-\alpha_{i+1})\cdot s_*(v) - p_i > \beta_i\cdot v - p_i = u_i$. This contradicts the assumption that~$s$ is an equilibrium. 

	Now assume that $s_*(v)\leq (\beta_i-\beta_{i+1})/(\alpha_i-\alpha_{i+1}) \cdot v$, and observe that in this case $s_*(v)<v$.  Let $v_\ell(\theta_\ell)=v$ for $\ell\leq i$ and $v_\ell(\theta_\ell)=0$ for $\ell>i$.  It then holds that $s_\ell(\theta_\ell)=s_*(v)>0$ for $\ell\leq i$ and $s_\ell(\theta_\ell)=s_*(0)=0$ for $\ell>i$.  Assume without loss of generality that for this bid profile, \aVCG assigns slot $\ell$ to agent $\ell$ for $\ell\leq i$, and that the remaining agents are not assigned a slot.  The utility of agent~$i$ under this assignment is $u_i = \beta_i\cdot v-p_i(\theta) = \beta_i\cdot v$.  If he changed his bid to $b>s_*(v)$, he would be assigned slot~$1$ at price $\sum_{j<i} (\alpha_{j}-\alpha_{j+1})\cdot s_*(v)$ and obtain utility $u_i' = \beta_1\cdot v-\sum_{j<i}(\alpha_{j}-\alpha_{j+1})\cdot s_*(v) > \beta_1\cdot v-\sum_{j<i}(\beta_{j}-\beta_{j+1})\cdot v = \beta_i\cdot v = u_i$. This again contradicts the assumption that~$s$ is an equilibrium. 

\section{Proof of \propref{prop:expostofvcg}}\label{app:expostofvcg}
%
	Consider two agents $i,i'\in N$ and two consecutive slots $j$ and $j+1$.  Fix $\theta_i$. Since the strategy $s_i$ of agent $i$ does not depend on the types of the other agents, we can choose the types of the agents in $N\setminus\{i,i'\}$ in such a way that an efficient outcome assigns slot $j$ to agents $i$ and slot $j+1$ to agent $i'$, or vice versa. Doing so also fixes the bids of all agents in $N\setminus\{i,i'\}$, and thus the payment $p_{j+1}$ associated with slot $j+1$. 
	
	For $s_i$ to be part of an ex-post equilibrium it must hold that
\begin{align*}
	& \alpha_j \cdot v_{i'}(\theta_{i'}) - (s_i(j,\theta_i)-s_i(j+1,\theta_i)) - p_{j+1} \geq \alpha_{j+1} \cdot v_{i'}(\theta_{i'}) - p_{j+1} 
	\intertext{if agent $i'$ is assigned slot $j$, and}
	& \alpha_{j+1} \cdot v_{i'}(\theta_{i'}) - p_{j+1} \geq \alpha_j \cdot v_{i'}(\theta_{i'}) - (s_i(j,\theta_i)-s_i(j+1,\theta_i)) - p_{j+1} 
\end{align*}
	if agent $i'$ is assigned slot $j+1$. For the equilibrium to be efficient the first inequality must hold if $v_{i'}(\theta_{i'})=v_{i}(\theta_i)+\epsilon$ for any $\epsilon>0$, and the second inequality must hold if $v_{i'}(\theta_{i'})=v_{i}(\theta_i)-\epsilon$ for any $\epsilon > 0$. By substituting $v_{i'}(\theta_{i'})$ accordingly in the two inequalities and considering arbitrarily small $\epsilon>0$ we obtain
\[
	s_i(j,\theta_i) = (\alpha_j - \alpha_{j+1}) \cdot v_i(\theta_i) + s_i(j+1,\theta_{i}) .
\]
	Since this equality must hold for every agent, the payment for slot $i$ equals
\begin{align*}
	p_i &= \sum_{j=i}^{\mathclap{\min(k,n-1)}} (s_{j+1}(j,\theta_{j+1})-s_{j+1}(j+1,\theta_{j+1})) 
	= \sum_{j=i}^{\mathclap{\min(k,n-1)}}((\alpha_j - \alpha_{j+1})\cdot v_{j+1}(\theta_{j+1})),	
\end{align*}
	which is exactly the VCG payment for that slot. 
	
	The stronger claim for $n>k$ follows by setting $\alpha_{k+1}=0$ and $s_{i}(k+1,\theta_i)=0$, and deriving strategies inductively from slot~$k$ through~$1$ according to the above equality. 

\section{Efficiency Requires an Exponential Number of Bids}
\label{app:exponential}

\begin{proposition} \label{prop:exponential}
	Consider a VCG combinatorial auction with a set $G$ of items, $|G|\geq 3$, and a set $N$ of agents, $|N|\geq 2$, and assume that $s=(s_1,\dots,s_n)$ is an efficient ex-post equilibrium. Then there exists a type profile $\theta=(\theta_1,\dots,\theta_n)$ such that for some agent $i\in N$ and every pair of bundles $X,Y\in 2^G$ such that $X\subset Y$, $s_i(X,\theta_i) \neq s_i(Y,\theta_i)$. 
\end{proposition}
\begin{proof}
	We can assume without loss of generality that $|N|=2$. Otherwise, for $i>2$, we could choose $\theta_i$ such that $v_i(X,\theta_i)=0$ for every $X\in 2^G$, which by efficiency of $s$ would immediately imply that $s_i(X,\theta_i)=0$ for all $X \subseteq 2^G$. 
	
	Choose $\theta_1$ such that $v_1(X,\theta_1)\neq v_1(Y,\theta_1)$ for all $X,Y\in 2^G$ with $X\subset Y$. Further assume for contradiction that there exist $X,Y\in 2^G$ such that $X\subset Y$ and $s_1(X,\theta_1)=s_1(Y,\theta_1)$. Observe that $X^c\supset Y^c$, and that $v_1(Y,\theta_1)>v_1(X,\theta_1)$ by assumption. 
	We can now choose $\theta_2$ such that $v_2(X^c,\theta_2) > v_2(Y^c,\theta_2)$ and $v_1(Y,\theta_1) + v_2(Y^c,\theta_2) > v_1(X,\theta_1) + v_2(X^c,\theta_2)$. By choosing $v_2(Y^c,\theta_2)$ large enough we can further ensure that the only efficient outcome assigns bundle $Y$ to agent~$1$ and bundle $Y^C$ to agent~$2$. 
	Now, since $s_1(X,\theta_1)=s_1(Y,\theta_1)$ and $v_2(X^c,\theta_2)>v_2(Y^c,\theta_2)$, $s_2$ can only be a best response to $s_1$ if the VCG mechanism assigns $X$ to agent~$1$ and $X^c$ to agent~$2$. This contradicts efficiency of $s$. 
\end{proof}

\section{Proof of Theorem~{\ref{thm:nVCG}}}
\label{app:nVCG}

We first prove that \nVCG is tight. By \propref{prop:tightness} it suffices to show that \nVCG satisfies outcome closure. Fix valuations $v_j$ and types $\theta_j$ for all $j \in N$, and consider an arbitrary agent $i \in N.$ We claim that for every $x_i \in X_i$ and every $\hx_{-i} \in \hX_{-i}$, there exists $\hx_i \in \hX_i$ such that
\[
    u_i((\hx_i,\hx_{-i}), \theta_i) \ge u_i((x_i, \hx_{-i}), \theta_i).
\] 
Denote the outcome of VCG for $(x_i,\hx_{-i})$ by $o$. Let $\hx_i(C) = x_i(C)$ for $C = o_i$ and let $\hx_i(C) = 0$ otherwise. We know that $o$ achieves the same social welfare under $(\hx_i,\hx_{-i})$ as under $(x_i,\hx_{-i}).$ Since $\hx_i(C) \le x_i(C)$ for all $C$, we also know that the social welfare achieved by any outcome $o' \neq o$ under $(\hx_i,\hx_{-i})$ is weakly smaller than the social welfare achieved by $o$ under $(\hx_i,\hx_{-i})$. Hence agent $i$ gets the same bundle of items, namely $o_i$, under $(\hx_i,\hx_{-i})$ and $(x_i,\hx_{-i}).$ Since $i$'s payment depends only on $\hx_{-i}$ we conclude that $u_i((\hx_i,\hx_{-i}), \theta_i) \ge u_i((x_i, \hx_{-i}), \theta_i)$. 

Now we turn to totality. To this end, we consider a property of a simplification we call outcome reducibility, which requires that there exists a way of mapping messages of the original mechanism to messages of the simplification such that (i)~outcomes and payoffs are preserved, and such that (ii)~for every agent and every choice of messages for the other agents, the maximum utility the agent can obtain, by choosing any action, is at least as high as the utility he can obtain by choosing from his restricted message set when the messages of the other agents have been mapped to their restricted message sets. More formally, a simplification $(N,\hX,\hf,\hp)$ of a mechanism $(N,X,f,p)$ satisfies \emph{outcome reducibility} if there exists a mapping $h:X\rightarrow\hX$ with the following properties: (i)~for every $x\in X$, $f(x) = f(h(x))$ and $p(x)=p(h(x))$; (ii)~for every $i\in N$, every $x \in X$, and every $\hx'_i\in\hX_i$, there exists $x'_i\in X_i$ such that $u_i((x'_i,x_{-i}),\theta_i)\geq u_i((\hx'_i,h_{-i}(x)),\theta_i)$. 
This turns out to be sufficient for totality, but only with respect to Nash equilibria. 
\begin{lemma} \label{lem:totality}
Every simplification that satisfies outcome reducibility is total with respect to Nash equilibria. 
\end{lemma}
\begin{proof}
Fix $\theta\in\Theta$.  Consider a mechanism $M=(N,X,f,p)$, a simplification $\hM=(N,\hX,\hf,\hp)$ that has property P2, and a Nash equilibrium $x$ of $M$.  Assume for contradiction that $h(x)\in\hX$ is not a Nash equilibrium of $\hM$.  Then, for some $i\in N$, there exists $\hx'_i\in\hX_i$ such that $u_i((\hat{x}'_i,h_{-i}(x)),\theta_i) > u_i(h(x),\theta_i)$.  Since $\hM$ has property P2, there further exists $x'_i \in X_i$ such that $u_i((x'_i,x_{-i}),\theta_i) \geq u_i((\hx'_i,h_{-i}(x)),\theta_i)$.  It follows 
that $u_i((x'_i,x_{-i}),\theta_i) > u_i(h(x),\theta_i) = u_i(x,\theta_i)$, which contradicts the assumption that $x$ is a Nash equilibrium of $M$. 
\end{proof}

By \lemref{lem:totality} it now suffices to show that $n$-VCG satisfies outcome reducibility. 
We compute $h(x)$ as follows: 1.~Denote the outcome of VCG for $x$ by $o$. For all agents $i$ and $B = o_i$ mark $x_i(B)$. 2.~For each agent $j$ let $o'$ denote the outcome of VCG if $j$ was removed. For each such outcome $o'$, all agents $i$, and $B = o'_i$ mark $x_i(B)$. 3.~For all agents $i$ and all bundles $B$ set $\hx_i(B) = x_i(B)$ if $x_i(B)$ was marked and $\hx_i(B) = 0$ otherwise.

To (i): The outcome that maximizes welfare under $x$ also maximizes welfare under $h(x)$ and the welfare achieved by this outcome is the same under $x$ and $h(x)$. Similarly, for each agent $i$, the welfare achieved by the outcome that maximizes welfare if agent $i$ is removed is the same under $x$ and $h(x)$. Hence the outcomes $f(h(x))$ and $f(x)$ and the prices $p(h(x))$ and $p(x)$ for $h(x)$ and $x$ are the same.

To (ii): Consider $\hx'_i.$ Denote the outcome computed by VCG for $(\hx'_i,h_{-i}(x))$ by $\hat{o}.$ We claim that there exists $x'_i$ such that in the outcome $o$ computed by VCG for $(x'_i,x_{-i})$ we have $o_i = \hat{o}_i.$ Let $x'_i(C) = \Wmax(\hx'_i, h_{-i}(x)) + \epsilon$ for $C = o_i$ and some $\epsilon > 0$ and let $x'_i(C) = 0$ otherwise. If $o_i = \hat{o}_i$, then the social welfare $\Wmax(x'_i, x_{-i})$ of the outcome of VCG for $(x'_i, x_{-i})$ is at least $\Wmax(\hx'_i,h_{-i}(x)) + \epsilon$. Otherwise, the social welfare is $\Wmax(x'_i, x_{-i}) = \Wmax(0,x_{-i}) = \Wmax(0,h_{-i}(x)) \le \Wmax(\hx'_i, h_{-i}(x))$ and, thus, strictly smaller. We conclude that $o_i = \hat{o}_i.$

We further claim that the price $p_i(o_i)$ that agent $i$ has to pay for $o_i$ under $(x'_i,x_{-i})$ is smaller/equal to the price $\hat{p}_i(o_i)$ that agent $i$ has to pay for $o_i$ under $(\hx'_i,h_{-i}(x))$. Since $\Wmax(0,x_{-i}) = \Wmax(0,h_{-i}(x))$ and $\Wmax^{o_i}(0,x_{-i}) \ge \Wmax^{o_i}(0,h_{-i}(x))$, where $\Wmax^{o_i}(\dots)$ denotes the maximal social welfare for the corresponding message profile if all items in $o_i$ are removed from the set of items, we have that $p_i(o_i) = \Wmax(0,x_{-i}) - \Wmax^{o_i}(0,x_{-i}) \le \Wmax(0,h_{-i}(x)) - \Wmax^{o_i}(0,h_{-i}(x)) = \hat{p}_i(o_i).$

We conclude that $u_i((x'_i,x_{-i}), \theta_i) \ge u_i((\hx'_i, h_{-i}(x)), \theta_i).$

\typeout{get arXiv to do 4 passes: Label(s) may have changed. Rerun}

\end{document}